\documentclass{article}

\usepackage[final]{nips_2016}

\usepackage[utf8]{inputenc} 
\usepackage[T1]{fontenc}    
\usepackage{hyperref}       
\usepackage{url}            
\usepackage{booktabs}       
\usepackage{amsfonts}       
\usepackage{nicefrac}       
\usepackage{microtype}      
\usepackage{graphicx}
\usepackage[numbers]{natbib}

\usepackage[usenames]{xcolor}
\usepackage{amssymb,amsmath,amsthm,amscd}
\usepackage{mathrsfs}
\usepackage{cleveref}
\usepackage{bbm}

\newtheorem{thm}{Theorem}[section]
\newtheorem{lem}[thm]{Lemma}

\newtheorem{prop}[thm]{Proposition}

\newtheorem{remark} [thm]{Remark}

\def \R {\mathbb R}

\newcommand{\diverg}{\operatorname{\mathrm{div}}}

\title{Data driven estimation of Laplace-Beltrami operator}

\author{
 Fr\'ed\'eric Chazal
\\
Inria Saclay\\
Palaiseau France\\
 \texttt{frederic.chazal@inria.fr} \\
 \And
Ilaria Giulini 
\\
Inria Saclay\\
Palaiseau France\\
\texttt{ilaria.giulini@me.com} \\
  \And
  Bertrand~Michel 
  \\
  Ecole Centrale de Nantes \\
  Laboratoire de Math\'ematiques Jean Leray (UMR 6629 CNRS)\\
  Nantes France \\ 
  \texttt{bertrand.michel@ec-nantes.fr} }

\begin{document}

\maketitle

\begin{abstract}
Approximations of Laplace-Beltrami operators on manifolds through graph Laplacians have become popular tools in data analysis and machine learning. These discretized operators usually depend on bandwidth parameters whose tuning remains a theoretical and practical problem. In this paper, we address this problem for the unnormalized graph Laplacian by establishing an oracle inequality that opens the door to a well-founded data-driven procedure for the bandwidth selection. Our approach relies on recent results by Lacour and Massart \cite{lacour2015minimal} on the so-called Lepski's method. 
\end{abstract}

\section{Introduction}

The Laplace-Beltrami operator is a fundamental and widely studied mathematical tool carrying a lot of intrinsic topological and geometric information about the Riemannian manifold on which it is defined. 
Its various discretizations, through graph Laplacians, have inspired many applications in data analysis and machine learning and led to popular tools such as  
Laplacian EigenMaps~\cite{belkin2003laplacian} for dimensionality reduction, spectral clustering ~\cite{von2007tutorial}, or semi-supervised learning~\cite{belkin2004semi}, just to name a few. 

During the last fifteen years, many efforts, leading to a vast literature, have been made to understand the convergence of graph Laplacian operators built on top of (random) finite samples to Laplace-Beltrami operators. For example pointwise convergence results have been obtained in \cite{belkin2005towards} (see also \cite{belkin2008towards}) and \cite{hein2007graph}, and a (uniform) functional central limit theorem has been established in \cite{gine2006empirical}. Spectral convergence results have also been proved by \cite{belkin2007convergence} and \cite{von2008consistency}. More recently, \cite{ting2011analysis} analyzed the asymptotic of a large family  of graph Laplacian operators by taking the diffusion process approach previously proposed in \cite{nadler2006diffusion}. 

Graph Laplacians depend on scale or bandwidth parameters whose choice is often left to the user.  Although many convergence results for various metrics have been established,  little is known about how to  rigorously and efficiently tune these parameters in practice.  In this paper we address this problem in the case of unnormalized graph Laplacian. More precisely, given a Riemannian manifold $M$ of known dimension $d$  and a function $f : M \to \mathbb R$ , we consider the standard unnormalized graph Laplacian operator defined by
\[
\hat \Delta_h f(y) = \frac{1}{nh^{d+2}} \sum_i K\left( \frac{y-X_i}{h}\right) \left[ f(X_i) - f(y) \right], \qquad y \in M,\]
where $h$ is a bandwidth,  $X_1,\ldots,X_n$ is a finite point cloud sampled on $M$ on which the values of $f$ can be computed, and $K$ is the Gaussian kernel: for $y \in \mathbb R^m$
\begin{equation}\label{gauss_k}
K(y) = \frac{1}{(4\pi)^{d/2}} e^{-\|y\|_m^2/4},  
\end{equation}
where $\|y\|_m$ is the Euclidean norm in the ambiant space $\R^m$. 

In this case, previous results (see for instance \cite{gine2006empirical}) typically say that the bandwidth parameter $h$ in $\hat \Delta_h$ should be taken of the order of $n^{-\frac{1}{d+2+\alpha}}$ for some $\alpha >0$, but in practice, for a given point cloud, these asymptotic results are not sufficient to choose $h$ efficiently.
In the context of neighbor graphs \cite{ting2011analysis} proposes self-tuning graphs  by choosing $h$ locally in terms of the distances to the $k$-nearest neighbor, but  note that $k$ still need to be chosen and moreover as far as we know there is no guarantee for such method to be rate-optimal. More recently a data driven method for spectral clustering  has been proposed in \cite{rieser2015topological}. Cross validation~\cite{arlot2010survey} is the standard approach for tuning parameters in statistics and machine learning. Nevertheless, the problem of choosing $h$ in  $\hat \Delta_h $ is not easy to rewrite as a cross validation problem, in particular because there is no obvious contrast corresponding to the problem (see \cite{arlot2010survey}).

The so-called Lepski's method is another popular method for selecting the smoothing parameter of an estimator. The method has been introduced by Lepski \cite{lepskii1992asymptotically,lepskii1993asymptotically,lepski1992problems} for kernel estimators and local polynomials for various risks and several improvements of the method have then been proposed, see~\cite{lepski1997optimal,goldenshluger2009structural,goldenshluger2008universal}. In this paper we adapt Lepski's method for selecting $h$ in the graph Laplacian estimator $\hat \Delta_h $. Our method is supported by mathematical guarantees: first we obtain an oracle inequality - see Theorem \ref{prop_oineq} -  and second we obtain the correct rate of convergence - see Theorem \ref{cor_c} - already proved in the asymptotical studies of  \cite{belkin2005towards} and \cite{gine2006empirical} for non data-driven choices of the bandwidth. Our approach follows the ideas recently proposed in \cite{lacour2015minimal}, but for the specific problem of Laplacian operators on  smooth manifolds. In this first work about the data-driven estimation of Laplace-Beltrami operator,  we focus as in \cite{belkin2005towards} and \cite{gine2006empirical}  on  the pointwise estimation problem:  we consider a smooth function $f$ on $M$ and the aim is to estimate $\hat \Delta f$ for the $L^2$-norm $\| \cdot\|_{2,M}$  on $M \subset \mathbb R^m$. The data driven method presented here may be adapted and generalized for other types of risks (uniform norms on functional family and convergence of the spectrum) and other types of graph Laplacian operators, this will be the subject of future works. 

The paper is organized as follows: Lepski's method is introduced in Section \ref{sec:Lep}. The main results are stated in Section \ref{res} and a sketch of their proof is given in Section \ref{pf_mthm} (the complete proofs are given in the supplementary material). A numerical illustration and a discussion about the proposed method are given in Sections \ref{sec:exp} and \ref{sec:disc} respectively. 


%
%
%
%
%
%
%

\section{Lepski's procedure for estimating the Laplace-Beltrami operator}
\label{sec:Lep}



All the Riemannian manifolds considered in the paper are smooth compact $d$-dimensional submanifolds (without boundary) of $\R^m$ endowed with the Riemannian metric induced by the Euclidean structure of $\R^m$. Recall that, given a compact $d$-dimensional smooth Riemannian manifold $M$ with volume measure $\mu$, its Laplace-Beltrami operator is the linear operator $\Delta$ defined on the space of smooth functions on $M$ as $\Delta(f) = - \diverg( \nabla f)$ where $\nabla f$ is the gradient vector field and $\diverg$ the divergence operator. In other words, using the Stoke's formula, $\Delta$ is the unique linear operator satisfying
$$\int_M \| \nabla f \|^2 d\mu = \int_M \Delta(f) f d\mu.$$

Replacing the volume measure $\mu$ by a distribution $\mathrm P$ which is absolutely continuous with respect to $\mu$, the weighted Laplace-Beltrami operator $\Delta_{\mathrm P}$ is defined as 
\begin{equation}\label{eqDL}
\Delta_{\mathrm P}f = \Delta f + \frac{1}{p} \langle \nabla p, \nabla f \rangle \, ,
\end{equation}
where $p$ is the density of $\mathrm P$ with respect to $\mu$. The reader may refer to classical textbooks such as, e.g., \cite{rosenberg1997laplacian} or \cite{grigoryan2009heat} for a general and detailed introduction to Laplace operators on manifolds. 

In the following, we assume that we are given n points  $X_1,\dots,X_n$ sampled on $M$ according to the distribution $\mathrm P$. Given a  smooth function $f$ on $M$, the aim is to estimate $\Delta_{\mathrm P} f$, by selecting an estimator in a given finite family of graph Laplacian $ (\hat \Delta_{h} f)_{h \in \mathcal H}$, where $\mathcal H$ is a finite family of bandwidth parameters.


Lepski's procedure is generally presented as a method for selecting bandwidth in an adaptive way.
More generally, this method can be seen as an estimator selection procedure. 

\subsection{Lepski's procedure}

We first shortly explain the ideas of Lepski's method. Consider a target quantity $s$, a collection of estimators $(\hat s_h)_{h\in\mathcal H}$ and a loss function $\ell(\cdot,\cdot)$. A  standard objective when selecting $\hat s_h$ is trying to minimize the risk $\mathbb E \ell( s,\hat s_h)$ among the family of estimators.  In most settings, the risk of an estimator can be decomposed into  a bias part and a variance part. Of course neither the risk, the bias nor the variance of an estimator are known in practice. However in many cases, the variance term can be controlled quite precisely. 
Lepski's method requires that the variance of each estimator $\hat s_h$ can be tightly  upper bounded by a quantity $v(h)$. In most cases, the bias can be written as $\ell (s,\bar s_h)$ where  $\bar s_h$ corresponds to some (deterministic) averaged version of $\hat s_h$. It thus seems natural to estimate  $\ell (s,\bar s_h)$ by $\ell (\hat s_{h'},\hat s_h)$ for some $h'$ smaller than $h$. The later quantity incorporates some randomness while the bias does not. The idea is to remove the ``random part" of the estimation by considering $\left[\ell (\hat s_{h'},\hat s_h) - v(h) - v(h')\right]_+$, where $[ \: ]_+$ denotes the positive part. The bias term is estimated by considering all pairs of estimators $( s_{h},\hat s_{h'})$  through the quantity $\sup_{h'\leq h}\left[\ell (\hat s_{h'},\hat s_h) - v(h) - v(h')\right]_+$. Finally, the estimator minimizing the sum of the estimated bias and variance is selected, see ~\cref{h_lepski} below.
 
In our setting, the control of the variance of the graph Laplacian estimators $\hat \Delta_h$ is not tight enough to directly apply the above described method. To overcome this issue, we use a more flexible version of Lepski's method that involves some multiplicative coefficients $a$ and $b$ introduced in the variance and bias terms.
More precisely, let $ V(h) = V_f(h)$ be an upper bound for $ \mathbb E[ \| (\mathbb E[\hat \Delta_h] - \hat \Delta_h)f \|_{2,M}^2]$. The bandwidth $\hat h$  selected by our Lepski's procedure is defined by 
\begin{equation}\label{h_lepski}
\hat h = \hat h_f= \mathrm{arg}\min_{h\in \mathcal H} \left\{  B(h) + b V(h) \right\}
\end{equation}
where 
\begin{equation}\label{bh}
 B(h) = B_f(h) = \max_{h'\leq h, \, h' \in \mathcal H} \left[ \|( \hat \Delta_{h'} - \hat \Delta_h )f\|_{2,M}^2 - a V(h')  \right]_+
\end{equation}
with $0<a\leq b$. The calibration of the constants $a$ and $b$ in practice is beyond the scope of this paper, but we suggest a heuristic procedure inspired from \cite{lacour2015minimal} in \cref{sec:exp}.

\subsection{Variance of the graph Laplacian for smooth functions}\label{ref_sec}

In order to control the variance term, we consider  for this paper the set $\mathcal F$ of smooth functions $f: M \to \mathbb R$ uniformly bounded up to the third order. For some constant $C_{\mathcal F}  >0 $ , let 
\begin{equation}\label{def_Cf}
\mathcal F = \left\{  f \in \mathcal C^3(M,\R) \,  , \,    \| f^{(k)}\|_{\infty} \leq C_{\mathcal F} , \,  {k=0,\dots,3} \right\}
\end{equation}
Here, by $\| f^{(k)} \|_\infty \leq C_{\cal F}$ we mean that in any normal coordinate systems all the partial derivatives of order $k$ of $f$ are bounded by $C_{\cal F}$. 

We introduce some notation before giving the variance term for $f \in \mathcal F$. Define
\begin{align}
\label{def_Da}
D_{\alpha} & =  \frac{1}{(4\pi)^{d}} \int_{\mathbb R^d} \left( \frac{C \| u\|_d^{\alpha+2}}{2} + C_1 \|u\|_d^{\alpha}  \right)\ e^{-\|u\|_d^2/4}   \, \mathrm du\\
\label{def_tDa}
\tilde D_{\alpha} & =  \frac{1}{(4\pi)^{d/2}} \int_{\mathbb R^d} \left( \frac{C \| u\|_d^{\alpha+2}}{4} + C_1 \|u\|_d^{\alpha}  \right)\ e^{-\|u\|_d^2/8}   \, \mathrm du
\end{align}
where $\| \|_d$ is the euclidean norm in $\R^d$ and where  $C$ and $C_1$ are geometric constants that only depend on the metric structure of $M$ (see  Lemma~\ref{lem_2.2} in the appendices). We also introduce the $d$-dimensional Gaussian kernel on $\R^d$: 
\[
K_d(u) = \frac{1}{(4\pi)^{d/2}} e^{-\|u\|_d^2/4}, \qquad u \in \mathbb R^d 
\]
and we denote by $\| \cdot\|_{p,d}$  the $L^p$-norm on $\mathbb R^d$.
The next proposition provides an explicit bound $V(h)$ on the variance term.
Let
\begin{equation}\label{od}
\omega_d = 3 \times 2^{d/2-1}
\end{equation}
and
\begin{equation}\label{ad}
\alpha_d(h) = h^2 \left(  D_4+\frac{3}{2} \omega_d \ \| K_d\|_{2,d}^2 + \frac{2\mu(M)}{(4\pi)^{d}} \right) 
+h^4 \frac{D_6}{4}.
\end{equation}

We first need to control the variance of $\hat \Delta_h f $ over $\cal F$. This will be possible by considering Taylor Young expansions of $f$  in normal coordinates. 
For that purpose, for technical reasons following from Lemma~\ref{lem_2.2}, we constrain the parameter $h$ to satisfy the following inequality
\begin{equation} 
\label{CondhmaxRayInj}
  2  \sqrt{d+4} h \log(h^{-1})^{1/2}  \leq  \rho(M) ,
\end{equation}
where $\rho(M)$ is a geometric constant that only depends on the reach and the injectivity radius of M.

\begin{prop}\label{eq_v}
Given $h\in \mathcal H$ and for any $f \in \mathcal F$, we have
\[
V(h) :=   \frac{2  C_{\mathcal F}^2}{n h^{d+2}} \Big[ w_d \| K_d\|_{2,d}^2 + \alpha_d(h)\Big] 
\leq  \mathbb E[ \| (\mathbb E[\hat \Delta_h] - \hat \Delta_h)f \|_{2,M}^2].
\]
\end{prop}

\noindent
For the proof we refer to ~\cref{proof1}.\\[1mm]


%
%

\section{Results}\label{res}

We now give the main result of the paper: an oracle inequality for the estimator $\hat \Delta_{\hat h}$, or in other words, a bound on the risk that shows that the performance of the estimator is almost
as good as it would be if we knew the risks of each estimator.
In particular it performs an (almost) optimal trade-off between the variance term $V(h)$ and
the approximation term
\begin{multline*}
D(h) = D_f(h) = \max\left\{  \| (p\Delta_{\mathrm P} -  \mathbb E [\hat \Delta_{ h}])f \|_{2,M}, \, \sup_{h'\leq h}   \| ( \mathbb E [\hat\Delta_{ h'}]-  \mathbb E [\hat\Delta_h])f  \|_{2,M} \right\}\\
\leq 2 \sup_{h'\leq h} \| (p\Delta_{\mathrm P} -  \mathbb E [\hat \Delta_{ h'}])f \|_{2,M}.
\end{multline*}


\begin{thm}
\label{prop_oineq}
According to the notation introduced in the previous section, 
let $\epsilon = \sqrt a/2 -1$ and 
\[
\delta(h) = \sum_{h'\leq h} \max\left\{  \exp\left( -\frac{\min\{\epsilon^2,\epsilon\}   \sqrt n }{24} \right), \,  \exp\left( -\frac{\epsilon^2}{3}\gamma_d(h') \right) \right\}
\]
and
\[
\gamma_d(h) = \frac{1}{ h^{d} \|p\|_\infty } \left[ \frac{ \omega_d \ \|K_d\|_{2,d}^2 + \alpha_d(h)}{\left(\omega_d\ \| K_d\|_{1,d} + \beta_d(h)\right)^2}   \right] 
\]
where $\alpha_d $ is defined by \eqref{ad} and where
\begin{align}
\label{bd}
\beta_d(h) & = h \left(\omega_d \ \| K_d\|_{1,d} +\frac{2 \mu(M)}{(4\pi)^{d/2}}  \right) +h^2\tilde D_3 + \frac{h^3 \tilde D_4}{2}.
\end{align}
Given $f \in \mathcal C^2(M,\R)$, with probability at least 
$1-2\sum_{h \in \mathcal H}\delta(h)$,
\[
\| (p\Delta_{\mathrm P}  - \hat \Delta_{\hat h}) f\|_{2,M} \leq \inf_{h \in \mathcal H} \left\{3 D(h) + (1+\sqrt2) \sqrt{b V(h)}\right\}. 
\]
\end{thm}

\noindent
Broadly speaking, Theorem~\ref{prop_oineq} says that there exists an event of large probability for which the estimator selected by Lepski's method is almost as good as the best estimator in the collection. Note that the size of the bandwidth family $\mathcal H$ has an impact on the probability term $1-2\sum_{h \in \mathcal H}\delta(h)$. If $\mathcal H$ is not too large, an oracle inequality for the risk of $\hat \Delta_{\hat h} f$ can be easily deduced from the later result. Henceforth we assume that $f \in \mathcal F$. We first  give a control on the approximation term $D(h)$. 

\begin{prop}\label{prop_Dh} 
Assume that the density $p$ is $\mathcal C^2$. 
It holds that 
\[
D(h) \leq \gamma \ C_{\mathcal F} h
\]
where $C_{\mathcal F}$ is defined in ~\cref{def_Cf} and $\gamma>0$ is a constant depending on $M$, $\| p\|_{\infty}$,  $\| p'\|_{\infty}$ and $\| p''\|_{\infty}$, where $\| p^{(k)} \|_\infty$ denotes the supremum of the absolute value of the partial derivatives of $p$ in any normal coordinates system. 
\end{prop}

We consider the following grid of bandwidths:
\[
\mathcal H = \left\{ e^{-k}\ , \ \lceil \log\log(n)\rceil \leq k \leq \lfloor \log(n) \rfloor  \right\}.
\]
Note that this choice ensures that Condition~\eqref{CondhmaxRayInj} is always satisfied for $n$ large enough.
The previous results lead to the pointwise rate of convergence of the graph Laplacian selected by Lepski's method:
\begin{thm}\label{cor_c}
Assume that the density $p$ is $\mathcal C^2$. For any $f \in \mathcal F$, we have
\begin{equation}
\label{cor_a1}
\mathbb E \left[ \| (p\Delta_{\mathrm P}  - \hat \Delta_{\hat h}) f\|_{2,M}  \right] \lesssim n^{-\frac{1}{d+4}}. 
\end{equation}
\end{thm}


\section{Sketch of the proof of ~\cref{prop_oineq}}\label{pf_mthm}
We observe that the following inequality holds
\begin{equation}\label{prop_base}
\|(p\Delta_{\mathrm P} - \hat \Delta_{\hat h})f \|_{2,M} \leq 
 D(h) + \|(\mathbb E[\hat \Delta_h] - \hat \Delta_h)f\|_{2,M} +  \sqrt{2\left(  B(h) + b V(h) \right)} . 
 \end{equation}
Indeed, for $h \in \mathcal H$,
\begin{align*}
\| (p\Delta_{\mathrm P} - \hat \Delta_{\hat h})f \|_{2,M}
& \leq \| (p\Delta_{\mathrm P} -  \mathbb E[\hat \Delta_h]  )f \|_{2,M} +  \| (\mathbb E[\hat \Delta_h]  - \hat \Delta_{ h})f \|_{2,M}
+ \| (\hat \Delta_h - \hat \Delta_{\hat h} )f \|_{2,M} \\
& \leq D(h) + \| (\mathbb E[\hat \Delta_h]  - \hat \Delta_{ h} )f\|_{2,M} 
+  \| (\hat \Delta_h - \hat \Delta_{\hat h} )f \|_{2,M} .
\end{align*}

\noindent
By definition of $B(h)$, for any $h'\leq h$,
\[
\| (\hat \Delta_{h'} - \hat \Delta_{h})f \|_{2,M}^2 \leq B(h) + a V(h') 
\leq  B( \max\{ h, h'\}) + a V( \min\{h,h'\}), 
\]
so that, 
according to the definition of $\hat h$ in ~\cref{h_lepski} and recalling that $a\leq b$,
\[
\| (\hat \Delta_{\hat h} - \hat \Delta_{h})f \|_{2,M}^2 \leq 2 \left[   B(h) +a V(h) \right] \leq 2 \left[   B(h) + b V(h) \right]
\]
which proves ~\cref{prop_base}. 

\vskip2mm
\noindent
We are now going to bound the terms that appear in ~\cref{prop_base}. 
The bound for $D(h)$ is already given in ~\cref{prop_Dh}, so that in the following we focus on
$B(h)$ and 
$\|(\mathbb E[\hat \Delta_h] - \hat \Delta_h)f\|_{2,M}$. 
More precisely the bounds we present in the next two propositions are 
based on the following lemma from \cite{lacour2015minimal}.

\begin{lem}\label{lemma1}  
Let $X_1, \dots, X_n$be an i.i.d. sequence of variables. 
Let $\widetilde{\mathcal S}$ a countable set of functions and let $\eta(s) = \frac{1}{n} \sum_i \left[ g_s(X_i) - \mathbb E[g_s(X_i)]  \right]$ for any $s \in \widetilde{\mathcal S}$. 
Assume that there exist constants $ \theta$ and $v_g$ such that for any $s\in \widetilde{\mathcal S}$
\[
\| g_s\|_{\infty} \leq \theta \quad \text{and} \quad \mathrm{Var}[ g_s(X)]\leq v_g.
\]
Denote $ H =  \mathbb E[ \sup_{s \in \widetilde{\mathcal S}} \eta(s) ]$. Then for any $\epsilon>0$ and any $H'\geq H$
\[
\mathbb P\left[ \sup_{s \in \widetilde{\mathcal S}} \eta(s) \geq (1+\epsilon)H'  \right] 
\leq \max \left\{   \exp\left( -\frac{\epsilon^2 n H'^2}{6v_g} \right), \,  \exp\left( -\frac{\min\{\epsilon^2,\epsilon\}  n H'}{24\ \theta} \right)  \right\}.
\]
\end{lem}

\vskip2mm
\begin{prop}\label{main_prop}
Let $\epsilon = \frac{\sqrt a}{2}-1$.
Given $h \in \mathcal H$, define 
\[
\delta_1(h) = \sum_{h'\leq h} \max\left\{  \exp\left( -\frac{\min\{\epsilon^2,\epsilon\}  \sqrt n }{24} \right), \,  \exp\left( -\frac{\epsilon^2}{3}\gamma_d(h') \right)  \right\}.
\]
With probability at least $1-\delta_1(h)$
\[
B(h) \leq 2 D(h)^2.
\]
\end{prop}


\begin{prop}\label{main_p2}
Let $\tilde \epsilon = \sqrt a -1$. 
Given $h \in \mathcal H$, define
\[
\delta_2(h) = \max\left\{  \exp\left( -\frac{\min\{\tilde \epsilon^2,\tilde \epsilon\}   \sqrt n }{24} \right), \,  \exp\left( -\frac{\tilde \epsilon^2}{3}\gamma_d(h) \right) \right\}.
\]
With probability at least $1-\delta_2(h) $
\[
\| (\mathbb E[\hat \Delta_h] - \hat \Delta_{ h})f \|_{2,M}  \leq \sqrt{a V(h)}.
\]
\end{prop}


\noindent
Combining the above propositions with ~\cref{prop_base}, we get that,  
for any $h\in \mathcal H$, with probability at least $1-(\delta_1(h) + \delta_2(h) )$,
\begin{align*}
\| (p\Delta_{\mathrm P} - \hat \Delta_{\hat h}) f\|_{2,M}
& \leq D(h) + \sqrt{a V(h)} + \sqrt{4 D(h)^2 + 2b V(h)}\\
& \leq 3 D(h) + (1+\sqrt2) \sqrt{b V(h)}
\end{align*}
where we have used the fact that $a\leq b$.
Taking a union bound on $h \in \mathcal H$ we conclude the proof.

\section{Numerical illustration} \label{sec:exp}

In this section we illustrate the results of the previous section on a simple example. In ~\cref{subsec:practical}, we describe a practical procedure when the data set $\mathbb{X}$ is sampled according to the uniform measure on $M$. A  numerical illustration us given in Section~\ref{subsec:sphere} when $M$ is the unit $2$-dimensional sphere in $\mathbb{R}^3$. 

\subsection{Practical application of the Lepksi's method} \label{subsec:practical}

Lepski's method presented in Section~\ref{sec:Lep} can not be directly applied in practice for two reasons. First, we can not compute the $L^2$-norm  $\| \, \|_{2,M}$ on $M$, the manifold $M$ being unknown. Second, the variance terms involved in Lepski's method are not completely explicit. 

Regarding the first issue, we can approximate $\| \, \|_{2,M}$  by splitting the data into two samples: an estimation sample  $\mathbb X_1$  for computing the estimators and a validation sample $\mathbb X_2$ for evaluating this norm. More precisely, given two estimators $ \hat \Delta_h f $ and $ \hat \Delta_{h'} f $ computed using $\mathbb X_1$, the quantity $   \| (\hat \Delta_h - \hat \Delta_{h'}) f \|_{2,M}^2  /  \mu(M) $ is approximated by the averaged sum  $\frac 1{n_2} \sum_{x \in \mathbb X_2} | \hat \Delta_h f(x)  - \hat \Delta_{h'} f (x) |^2$, where $n_2$ is the number of points in $\mathbb X_2$. We use these approximations to evaluate the bias terms  $B(h)$ defined by \eqref{bh}. 

The second issue comes from the fact that the variance terms involved in Lepski's method  depend  on the metric properties of the manifold and on the sampling density, which are both unknown. Theses variance terms are thus only known up to a multiplicative constant. This situation contrasts with more standard frameworks for which a tight and explicit control on the variance terms can be proposed, as in  \citep{lepskii1992asymptotically,lepskii1993asymptotically,lepski1992problems}. To address this second issue, we follow the calibration strategy recently proposed in \cite{lacour2015minimal} (see also \cite{lacour2016estimator}). In practice we remove all the multiplicative constants from  $V(h)$: all these constants are passed into the terms a and b. This means that we rewrite Lepski's method as follows:
\begin{equation*} 
\hat h(a,b)  = \mathrm{arg}\min_{h\in \mathcal H} \left\{  B(h) + b \frac{1}{n h ^4} \right\}
\end{equation*}
where 
\[
 B(h) = \max_{h'\leq h, \, h' \in \mathcal H} \left[ \|( \hat \Delta_{h'} - \hat \Delta_h )f\|_{2,M}^2 - a \frac{1}{n {h'} ^4}  \right]_+ .
\]
We choose $a$ and $b$ according to the following heuristic: 
\begin{enumerate}
\item Take $b=a$ and consider the sequence of selected models: $\hat h(a,a)$,
\item Starting from large values of $a$, make $a$ decrease and find the location $a_0$ of the main {\it bandwidth jump} in the step function $a \mapsto \hat h(a,a)$,
\item Select the model $\hat h(a_0,2a_0)$.
\end{enumerate}
The justification of this calibration method is currently the subject  of mathematical studies (\cite{lacour2015minimal}). Note that a similar strategy called "slope heuristic" has been proposed for calibrating $\ell_0$ penalties in various settings by strong mathematical results, see for instance \cite{birge2007minimal,arlot2009data,baudry2012slope}.

\subsection{Illustration on the sphere} \label{subsec:sphere}

In this section we illustrate the complete method on a simple example with data points generated uniformly on the sphere $\mathbb S^2$ in $\R^3$. In this case, the weighted Laplace-Beltrami operator is equal to the (non weighted) Laplace-Beltrami operator on the sphere. 

We consider the function $f(x,y,z) = (x^2 + y ^2 + z )  \sin x  \cos x $.  
The restriction of this function on the sphere has the following representation in spherical coordinates:
$$ \tilde f(\theta,\phi) = (\sin ^2  \phi   + \cos \phi )  \sin( \sin \phi \cos \theta) \cos( \sin \phi \cos \theta) .$$
It is well known that the Laplace-Beltrami operator on the sphere satisfies (see  Section~3 in \cite{grigoryan2009heat}):
$$
\Delta_{\mathcal S^2}   u= \frac 1 {\sin^2 \phi} \frac{\partial^2    u }{\partial  \theta^2 } + \frac{1}{\sin \phi }  \frac{\partial  }{\partial  \phi } \left(  \sin \phi   \frac{\partial   u  }{\partial  \phi}  \right) 
$$
for any smooth polar function $u$. This allows us to derive an analytic expression of   $\Delta_{\mathcal S^2} \tilde f$. 

We sample $n_1= 10^6$ points on the sphere for computing the graph Laplacians and we use  $n=10^3$ points for approximating the norms $ \| (\hat \Delta_h - \hat \Delta_{h'}) \tilde f  \|_{2,M}^2$. We compute the graph Laplacians for bandwidths in a grid $\mathcal H$ between 0.001 and 0.8 (see \cref{fig:GraphLap}). The risk of each graph Laplacian is estimated by a standard Monte Carlo procedure (see \cref{fig:risk}). 
\begin{figure}[h]
\centering
	\includegraphics[scale=0.3]{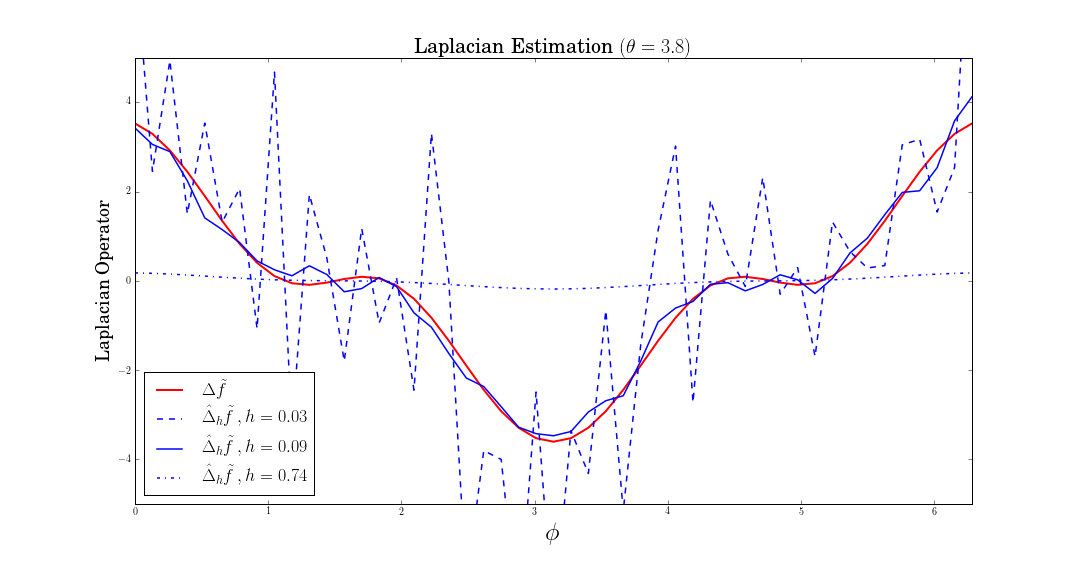}
\caption{Choosing $h$ is crucial for estimating $\Delta_{\mathcal S^2} \tilde f$:  small bandwidth overfits  $\Delta_{\mathcal S^2} \tilde f$ whereas  large bandwidth leads to almost constant  approximation functions of $\Delta_{\mathcal S^2} \tilde f$.}
\label{fig:GraphLap}
\end{figure}
\begin{figure}[h]
\centering
	\includegraphics[scale=0.3]{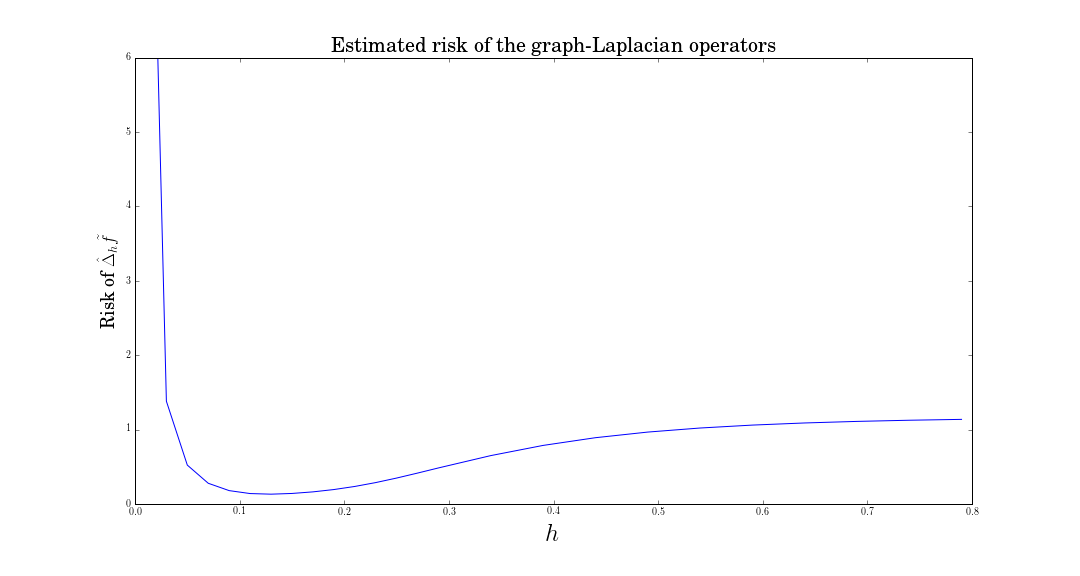}
\caption{Estimation of the risk of each graph Laplacian operator: the oracle Laplacian is for approximatively  $h=0.15$.}
\label{fig:risk}
\end{figure}


 
Figure~\ref{fig:jump} illustrates the calibration method. On this picture, the $x$-axis corresponds to the values of $a$ and the $y$-axis represents the bandwidths. The blue step function represents the function $a \mapsto \hat h(a,a)$. The red step function gives the model selected by the rule $a \mapsto \hat h(a,2 a)$. Following the heuristics given in Section~\ref{subsec:practical}, one could take for this example the value $a_0 \approx 3.5$ (location of the bandwidth  jump for the blue curve) which leads to select the model $\hat h(a_0,2a_0) \approx 0.2$ (red curve).

\begin{figure}[h]
\centering
	\includegraphics[scale=0.3]{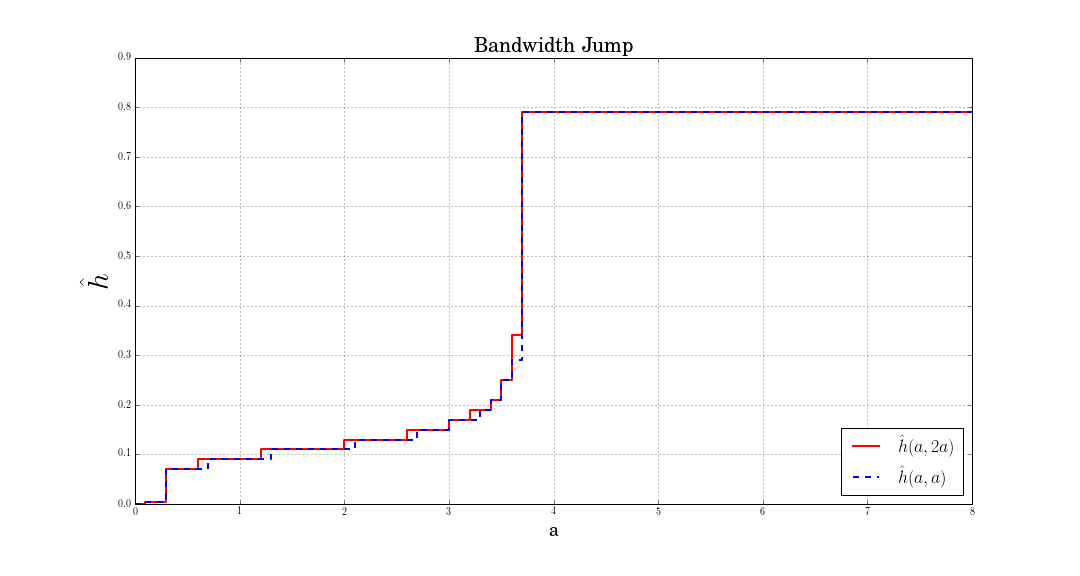}
\caption{Bandwidth jump heuristic: find the location of the jump (blue curve) and deduce the selected bandwidth with the red curve.}
\label{fig:jump}
\end{figure} 
 


\section{Discussion} \label{sec:disc}

This paper is a first attempt for a complete and well-founded data driven method for inferring Laplace-Beltrami operators from data points. Our results suggest various extensions and raised some questions of interest.  
%
%
For instance, other versions of the graph Laplacian have been studied in the literature (see  for instance \cite{hein2007graph,belkin2008towards}), for instance when data is not sampled uniformly. It would be relevant to propose a bandwidth selection method for these alternative estimators also.

From a practical point of view, as explained in ~\cref{sec:exp},  there is a gap between the theory we obtain in the paper and what can be done in practice. To fill this gap, a first objective is to prove an oracle inequality in the spirit of Theorem~\ref{prop_oineq} for a bias term defined in terms of the empirical norms computed in practice. A second objective is to propose mathematically well-founded heuristics for the calibration of the parameters $a$ and $b$.


Tuning bandwidths for the estimation of the spectrum of the Laplace-Beltrami operator is a difficult but important problem in data analysis. We are currently working on the adaptation of our results to the case of operator norms and spectrum estimation.


\section*{Appendix: the geometric constants $C$ and $C_1$} \label{sec:geomConst}

The following classical lemma (see, e.g. \cite{gine2006empirical}[Prop. 2.2 and Eq. 3.20]) relates the constants $C$ and $C_1$ introduced in Equations \eqref{def_Da} and \eqref{def_tDa} to the geometric structure of $M$. 

\begin{lem} \label{lem_2.2}
There exist constants $C, C_1 >0$ and a positive real number $r>0$ such that for any $x \in M$, and any $v \in T_x M$ such that $\| v \| \leq r$, 
\begin{equation}\label{eq_310}
\left|\sqrt{\mathrm{det}(g_{ij})}(v) - 1 \right| \leq  C_1 \| v\|_d^2 
\qquad \text{and}
\qquad \frac{1}{2} \| v\|_d^2 \leq  \| v\|_d^2 - C  \| v\|_d^4 
\leq \|\mathcal E_x(v)- x \|_m^2 \leq  \| v\|_d^2
\end{equation}
where $\mathcal E_x: T_xM \to M$ is the exponential map and $(g_{i,j})_{i,j} \in \{1,\cdots, d\}$ are the components of the metric tensor in any normal coordinate system around $x$. 
\end{lem}

Although the proof of the lemma is beyond the scope of this paper, notice that one can indeed give explicit bounds on $r$ and $C$ in terms of the reach and injectivity radius of the submanifold $M$. 

\subsubsection*{Acknowledgments} The authors are grateful to Pascal Massart for helpful discussions on Lepski's method and to Antonio Rieser for his careful reading and for pointing a technical bug in a preliminary version of this work. This work was supported by the ANR project TopData ANR-13-BS01-0008 and ERC Gudhi No. 339025

{\small
\bibliographystyle{alpha}
\bibliography{LB}
}

\newpage
\section*{Appendix: Proofs}

\noindent
For the sake of simplicity, we introduce the 
renormalized kernel
\[
H_h(y) = \frac{1}{h^{d+2}} K\left( {y}/{h}\right) = \frac{1}{h^{d+2} (4\pi)^{d/2}} e^{-\|y\|_m^2/4h^2}, \qquad y \in \mathbb R^m,
\] 
so that $\hat \Delta_h$ rewrites as 
\[
\hat \Delta_h f(y) =  \frac{1}{n}  \sum_{i=1}^n  H_h(y-X_i)  \left[ f(X_i) - f(y) \right]
\]
where we recall that  $X_1,\ldots,X_n$ is a finite point cloud (i.i.d.) sampled on $M$.
Note that the expectation $\Delta_h$ of $\hat \Delta_h$ satisfies
\[
\Delta_hf(y) = \mathbb E \hat \Delta_h f(y) = \int  H_h(y-x)  \left[ f(x) - f(y) \right] \, \mathrm d \mathrm P(x).
\]

\vskip2mm
\noindent
We present a technical result that is useful in the following and 
whose proof is postpone 
to ~\cref{pf_lem62}. 

\begin{lem}\label{tec_lem1}
Let $x \in M$ and $h\in \mathcal H$. 
According to the notation introduced in ~\cref{ref_sec} and in~\cref{res}
\begin{equation}
\label{tec_eq1}
\int_M \left| H_h(y-x) (f(x) - f(y)) \right| \, \mathrm d \mu(y) 
\leq \frac{C_{\mathcal F}}{h} \Big[ w_d \| K_d\|_{1,d} + \beta_d(h)\Big] =:\mathcal I_1(h) 
\end{equation}
Moreover
\begin{equation}
\label{tec_eq2}
\int_M \left| H_h(y-x) (f(x) - f(y)) \right|^2 \, \mathrm d \mu(y) 
\leq \frac{2C_{\mathcal F}^2}{h^{d+2}} \Big[ w_d \| K_d\|_{2,d}^2 + \alpha_d(h)\Big]=:\mathcal I_2(h).
\end{equation}
\end{lem}

\vskip2mm
\noindent
We also observe the following facts. 

\begin{remark}\label{nb_exp}
Given $t\in \mathbb R$ and $s,q>0$,
\[
t^q e^{-t^2/s^2} \leq \lambda e^{-t^2/2s^2}
\]
where $\lambda= s^q q^{q/2} e^{-q/2}.$
\end{remark}

\begin{remark}\label{nb_2}
Given $a,b >0$,
\[
e^{-(a-b)} - e^{-a} \leq b e^{-a/2}
\] 
so that,  under the assumptions of ~\cref{lem_2.2}, given $s>0$,
\begin{align*}
e^{-\|\mathcal E_x(v) - x\|_m^2/s^2} - e^{-\|v\|_d^2/s^2} 
& \leq e^{-(\|v\|_d^2- C  \| v\|_d^4)/s^2} - e^{-\|v\|_d^2/s^2} \leq \frac{C  \| v\|_d^4}{s^2}e^{-\|v\|_d^2/2s^2}.
\end{align*}
\end{remark}

\subsection{Proof of ~\cref{eq_v}}\label{proof1}
In order to get the bound $V(h)$ for $ \mathbb E[ \| (\mathbb E[\hat \Delta_h] - \hat \Delta_h)f \|_{2,M}^2]$
we observe that, according to the definition of $\hat \Delta_h$ and using the fact that the sample is i.i.d., 
\[
\mathbb E[ \| (\mathbb E[\hat \Delta_h] - \hat \Delta_h)f \|_{2,M}^2]
 \leq \frac{1}{n} \mathbb E\left[ \int \left|H_h(y-X)(f(X)- f(y) \right|^2\, \mathrm d\mu(y) \right]
\leq  \frac{1}{n}\ \mathcal I_2(h)
\]
where the last line follows from ~\cref{tec_lem1}.

\subsection{Proof of ~\cref{prop_oineq}}
According to the sketch of the proof provided in ~\cref{pf_mthm}, we only need to prove 
~\cref{main_prop} and ~\cref{main_p2}.

\subsubsection{Proof of ~\cref{main_prop}}\label{proof3}
Recalling the definition of $B(h)$, since, 
for any $h'\leq h$, 
\begin{align*}
\| (\hat \Delta_{h'} - \hat \Delta_h)f\|_{2,M}^2
& \leq 2\bigg[  \| (\hat \Delta_{h'} -\Delta_{h'}+\Delta_h-  \hat \Delta_h)f \|_{2,M}^2 +
 \| ( \Delta_{h'}-\Delta_h)f\|_{2,M}^2 \bigg]\\
& \leq 2\bigg[ \| (\hat \Delta_{h'} -\Delta_{h'}+\Delta_h-  \hat \Delta_h)f \|_{2,M}^2 + D(h)^2 \bigg]
\end{align*}
we get
\[
B(h) \leq \max_{h'\leq h} \left[ 2\| (\hat \Delta_{h'} -\Delta_{h'}+\Delta_h-  \hat \Delta_h)f \|_{2,M}^2 + 2D(h)^2 - a V(h') \right]_+.
\]
Thus it is sufficient to prove that 
\[
2 \| (\hat \Delta_{h'} -\Delta_{h'}+\Delta_h-  \hat \Delta_h)f\|_{2,M}^2 \leq
a V(h').
\]
We can write
\begin{align*}
\| (\hat \Delta_{h'} -\Delta_{h'}+\Delta_h-  \hat \Delta_h)f\|_{2,M}
= \sup_{\substack{t \in B_{2,M}(1)}} \eta(f,t)
\end{align*}
where $B_{2,M}(1) = \left\{ t \in L^2(M) \ | \ \|t\|_{2,M}=1  \right\}$,
\[
\eta(f,t) =  \langle t, (\hat \Delta_{h'} -\Delta_{h'}+\Delta_h-  \hat \Delta_h) f \rangle = \frac{1}{n}\sum_{i=1}^n g_{t,f}(X_i) - \mathbb E[g_{t,f}(X)]
\]
and
\[
g_{t,f}(x) =   \int  t(y) \left[ H_{h'}(y- x) - H_h(y- x)   \right] (f(x) - f(y) ) \, \mathrm d \mu(y).
\]

\noindent
Since the function $t \mapsto \eta(f,t) $ is continuous on $B_{2,M}(1)$, 
we consider $t \in {\mathcal T}$ where
${\mathcal T}$ is a countable set of $B_{2,M}(1)$.
In order to apply ~\cref{lemma1}
we need to compute the quantities $\theta,$ $v_g$ and $H$.\\[1mm]

\begin{lem}\label{lem65}
Let $t \in {\mathcal T}.$ With the same notation as in~\cref{tec_lem1} 
\begin{align*}
 \|g_{t,f}\|_{\infty} 
& \leq 2 \ \mathcal I_2(h')^{1/2} =: \theta \\[1mm]
\mathrm{Var}[g_{t,f}(X)]  
& \leq  4 \| p\|_\infty \ \mathcal I_1(h')^{2}=: v_g \\[1mm]
\mathbb E\left[ \sup_{t\in {\mathcal T}} \eta(f,t) \right] 
& \leq \frac{2 }{\sqrt n}\ \mathcal I_2(h')^{1/2}= : H.
\end{align*}
\end{lem}

\begin{proof}
Using the Cauchy-Schwarz inequality and the fact that $t \in B_{2,M}$ we get
\begin{multline*}
|g_{t,f}(x)| \leq \| t\|_{2,M} \left( \int \Big| \left( H_{h'}(y -   x) - H_h(y - x)   \right) (f(x) - f(y) )\Big|^2 \, \mathrm d \mu(y)\right)^{1/2}\\
\leq \Bigg( 2 \int \Big| H_{h'}(y -   x) (f(x) - f(y) )\Big|^2 \, \mathrm d \mu(y)\\
+ 2 \int \Big| H_h(y - x) (f(x) - f(y) )\Big|^2 \, \mathrm d \mu(y)   \Bigg)^{1/2}.
\end{multline*}
According to ~\cref{tec_lem1} and recalling that $h' \leq h$ we get
\begin{align*}
|g_{t,f}(x)| &\leq  \big( 2 \mathcal I_2(h') + 2 \mathcal I_2(h) \big)^{1/2} \leq 2 \ \mathcal I_2(h')^{1/2}
\end{align*}
which proves the first bound. To prove the second inequality we observe that, by the Cauchy-Schwarz inequality, 
\begin{multline*}
\mathrm{Var}[g_{t,f}(X)] 
 \leq \mathbb E \left[ g_{t,f}(X)^2 \right]\\
 = \mathbb E \left[ \left( \int  t(y) \left[  (H_{h'} - H_h)(y-X)  \right] (f(X) - f(y) ) \, \mathrm d \mu(y)  \right)^2 \right]\\
\leq  \mathbb E \Bigg[
\left| \int t(y)^2 \left( (H_{h'} - H_h)(y-X) \right)(f(X) - f(y) ) \, \mathrm d\mu(y) \right| \times \\
\times \left| \int \left( (H_{h'} - H_h)(y-X) \right) (f(X) - f(y) ) \, \mathrm d\mu(y) \right|
\Bigg].
\end{multline*}
According to ~\cref{tec_lem1}
\[
\left| \int \left( (H_{h'} - H_h)(y-X) \right) (f(X) - f(y) ) \, \mathrm d\mu(y) \right|
\leq  \mathcal I_1(h') + \mathcal I_1(h) \leq 2 \  \mathcal I_1(h')
\]
so that, recalling that the distribution $\mathrm P$ has a density $p$ with respect to $\mu$, 
\begin{align*}
\mathrm{Var}[g_{t,f}(X)] 
& \leq 2 \  \mathcal I_1(h') \ \mathbb E \Bigg[
\left| \int t(y)^2 \left( (H_{h'} - H_h)(y-X) \right)(f(X) - f(y) ) \, \mathrm d\mu(y) \right| \Bigg]\\
& \leq 2 \  \mathcal I_1(h') \ 
\int \left|  t(y)^2 \left( (H_{h'} - H_h)(y-X) \right)(f(X) - f(y) ) \right| \ p(x)  \, \mathrm d\mu(x) \mathrm d\mu(y)\\
& \leq 2 \| p\|_\infty \  \mathcal I_1(h') \ 
\int t(y)^2 \left( \int  \left| (H_{h'} - H_h)(y-X)(f(X) - f(y) ) \right|\, \mathrm d\mu(x) \right) \, \mathrm d\mu(y).
 \end{align*}
 Using again ~\cref{tec_lem1} and the fact that $t \in B_{2,M}$ we conclude that 
 \[
 \mathrm{Var}[g_{t,f}(X)] \leq 4 \| p\|_\infty \  \mathcal I_1(h')^2.
 \] 
 To get the last inequality we follow the proof of~\cref{eq_v}. We observe that 
 \begin{multline*}
 \mathbb E\left[ \sup_{t\in {\mathcal T}} \eta(f,t)^2 \right]
 =  \mathbb E \left[  \left\| (\hat \Delta_{h'} -  \hat \Delta_h)f - \mathbb E\left[ (\hat \Delta_{h'} -  \hat \Delta_h)f\right] \right\|_{2,M}^2\right] \\
 \leq \frac{1}{n}  \mathbb E \left[  \int \left| (H_{h'} - H_h)(y-X)(f(X) - f(y) ) \right|^2\, \mathrm d\mu(y) \right].
 \end{multline*}
According to ~\cref{tec_lem1},
\[
  \mathbb E\left[ \sup_{t\in {\mathcal T}} \eta(f,t)^2 \right]
  \leq \frac{2 }{n} \left( \mathcal I_2(h') +  \mathcal I_2(h) \right) \leq   \frac{4}{n} \mathcal I_2(h') 
\]
  which concludes the proof.
\end{proof}

\vskip2mm
\noindent
According to  ~\cref{lemma1},
with probability at least 
$1- \max\left\{  \exp\left( -\frac{\epsilon^2 n H^2}{6v_g} \right), \,  \exp\left( -\frac{\min\{\epsilon^2,\epsilon\}   n H}{24 \theta} \right) \right\}$, 
\begin{align*}
\|( \hat \Delta_{h'} -\Delta_{h'}+\Delta_h-  \hat \Delta_h )f\|_{2,M} & \leq (1+\epsilon) H
\end{align*}
where
\[
\frac{H}{\theta} = \frac{1}{\sqrt n} \qquad  \text{and} \qquad 
\frac{H^2}{ v_g}  = \frac{\mathcal I_2(h')}{n \mathcal I_1(h')} 
= \frac{2}{n\|p\|_\infty h'^{d}  } \left[ \frac{ \omega_d \ \|K_d\|_{2,d}^2 + \alpha_d(h')}{\left(\omega_d\ \| K_d\|_{1,d} + \beta_d(h')\right)^2}   \right] = \frac{2\gamma_d(h')}{n}.
\]

\vskip2mm
\noindent
Moreover, by definition, 
$H^2 \leq 4 V(h')$, so that 
choosing $\epsilon$ such that $a\geq 4 (1+\epsilon)^2$
we get that 
with probability at least $1- \max\left\{  \exp\left( -\frac{\min\{\epsilon^2,\epsilon\}  \sqrt n }{24} \right), \,  \exp\left( -\frac{\epsilon^2}{3}\gamma_d(h') \right) \right\}$
\begin{align*}
2\|( \hat \Delta_{h'} -\Delta_{h'}+\Delta_h-  \hat \Delta_h )f\|_{2,M}^2 
& \leq 4(1+\epsilon)^2 V(h')  \leq a V(h').
\end{align*}

\vskip2mm
\noindent
In particular we choose $\epsilon = \sqrt{a}/2 -1$.
The result follows taking a union bound on $h'$.

\subsubsection{Proof of ~\cref{main_p2}}\label{proof4}
The proof follows the one of ~\cref{main_prop}. 
We can write
\[
\| (\Delta_h - \hat \Delta_{ h})f \|_{2,M} 
= \sup_{t\in B_{2,M}(1)} \tilde \eta(f,t)
\]
where
\[
 \tilde \eta(f,t) = \frac{1}{n} \sum_{i=1}^n \Big( \tilde g_{t,f}(X_i) - \mathbb E [ \tilde g_{t,f}(X_i)] \Big)
\]
and
\[
 \tilde g_{t,f}(x) = \int t(y) H_h(y-x) (f(x)- f(y) ) \, \mathrm d \mu(y).
\]
Moreover we observe that we can consider $t \in \mathcal T$ where $\mathcal T \subset B_{2,M}(1)$ is a countable set.
\begin{lem}
Let $t \in {\mathcal T}.$ With the notation of ~\cref{tec_lem1}

\begin{align*}
 \|\tilde g_{t,f}\|_{\infty} 
& \leq \mathcal I_2(h)^{1/2}=: \tilde \theta\\[1mm]
\mathrm{Var}[\tilde g_{t,f}(X)]  
& \leq \|p\|_\infty \ \mathcal I_1(h)^2 =: \tilde {v_g}\\[1mm]
\mathbb E\left[ \sup_{t\in {\mathcal T}} \tilde \eta(f,t) \right] 
& \leq \frac{1}{\sqrt n} \ \mathcal I_2(h)^{1/2}= : \tilde H.
\end{align*}
\end{lem}

\begin{proof}
We proceed as in the proof of ~\cref{lem65}.
By the Cauchy-Schwarz inequality and ~\cref{tec_lem1}, recalling that $t \in B_{2,M},$ we get
\[
| \tilde g_{t,f}(x) |
\leq \| t\|_{2,M} \left( \int \left| H_h(y-x) (f(x)- f(y) ) \right|^2\, \mathrm d \mu(y)\right)^{1/2}
\leq \mathcal I_2(h)^{1/2}
\]
which proves the first inequality. To get the second one, we observe that
\begin{multline*}
\mathrm{Var}[\tilde g_{t,f}(X)] 
 \leq \mathbb E \left[\tilde g_{t,f}(X)^2 \right]\\
  \leq  \mathbb E \Bigg[
\left| \int t(y)^2  H_h(y-X) (f(X) - f(y) ) \, \mathrm d\mu(y) \right| \\
\times \left| \int H_h(y-X) (f(X) - f(y) ) \, \mathrm d\mu(y) \right|
\Bigg]\\
\leq \|p\|_\infty \ \mathcal I_1(h)^2.
\end{multline*}
To prove the last inequality we observe that 
\begin{multline*}
\mathbb E\left[ \sup_{t\in {\mathcal T}} \tilde \eta(f,t)^2 \right] 
 = \mathbb E \left[  \left\|   \hat \Delta_hf - \mathbb E\left[ \hat \Delta_h f\right] \right\|_{2,M}^2\right] \\
 \leq \frac{1}{n} \mathbb E \left[ \int \left| H_h(y-X) (f(X)-f(y) ) \right|^2\, \mathrm d\mu(y) \right] 
 \leq \frac{1}{n} \ \mathcal I_2(h).
\end{multline*}
\end{proof}

\vskip2mm
\noindent
According to ~\cref{lemma1},
with probability at least $1- \max\left\{  \exp\left( -\frac{\tilde \epsilon^2 n \tilde H^2}{6 \tilde v_g} \right), \,  \exp\left( -\frac{\min\{\tilde \epsilon^2,\tilde \epsilon\}   n \tilde H}{24 \ \tilde \theta} \right) \right\}$, we have
\begin{align*}
\| (\Delta_h - \hat \Delta_{ h})f \|_{2,M}  & \leq (1+\tilde \epsilon) \tilde H
\end{align*}
where
\[
\tilde H/\tilde \theta = 1/\sqrt n \qquad  \text{and} \qquad 
\tilde H^2/ \tilde v_g  = \frac{\mathcal I_2(h)}{n\|p\|_\infty \mathcal I_1(h)}= \frac{2}{n \| p\|_{\infty} h^{d}  } \left[ \frac{ 2\omega_d \ \|K_d\|_{2,d}^2 + \alpha_d(h)}{\left(2\omega_d\ \| K_d\|_{1,d} + \beta_d(h)\right)^2}   \right] = \frac{2\gamma_d(h)}{n}.
\]
Moreover, since $\tilde H^2 = V(h)$, 
choosing $\epsilon$ such that $a\geq (1+\tilde \epsilon)^2$
we get that 
with probability at least $1- \max\left\{  \exp\left( -\frac{\min\{\tilde \epsilon^2,\tilde \epsilon\}   \sqrt n }{24} \right), \,  \exp\left( -\frac{\tilde \epsilon^2}{3}\gamma_d(h) \right) \right\}$
\begin{align*}
\| (\Delta_h - \hat \Delta_{ h})f \|_{2,M}^2  \leq a V(h).
\end{align*}

\vskip2mm
\noindent
In particular choosing $\tilde \epsilon = \sqrt a -1$ we conclude the proof.

\subsection{Proof of ~\cref{prop_Dh}}\label{proof7}

We first prove that,
given $x \in M$, 
\begin{equation}\label{monster}
| p(x) \Delta_{\mathrm P}f(x) -  \Delta_{ h'}f(x) | 
\leq \Phi_{\mathcal F}(h')
\end{equation}
where
\begin{multline*}
\Phi_{\mathcal F} (h') =  
 h' \ C_{\mathcal F} \left[ 5\tau_d \left( \frac{\| p\|_\infty}{3} +\|p'\|_\infty + \|p''\|_\infty   \right) + \|p\|_\infty \tilde D_3 + \frac{2}{(4\pi)^{d/2}} \right] \\
 +  h'^2 \ C_{\mathcal F}\left[ \frac{35}{4} \|p''\|_\infty \tau_d+ \left( \|p'\|_\infty + \tilde D_4\frac{\| p\|_\infty}{2} \right) \right] \\
 + \frac{1}{2} h'^3 \ C_{\mathcal F} \tilde D_5 \left( \frac{\| p\|_\infty}{3} +\|p'\|_\infty + \|p''\|_\infty   \right)
 + \frac{1}{4} h'^4\ C_{\mathcal F} \|p''\|_\infty \tilde D_6\\
 +\frac{  \kappa_d \ C_{\mathcal F} \left( \|p'\|_\infty + \| p\|_\infty/2 \right)}{(4\pi)^{d/2}} (4(d+3))^{d/2} \log(h'^{-1})^{d/2} h'^{d+3} \\
 +  \frac{ \kappa_{d-1} \ C_{\mathcal F} \| p\|_\infty}{(4\pi)^{d/2}} (4(d+3))^{(d-1)/2} \log(h'^{-1})^{(d-1)/2} h'^{d+2}
\end{multline*}
with
\begin{align*}
\tau_d &= \frac{1}{(4\pi)^{d/2}} \int_{\mathbb R^d} e^{-\|u\|_d^2/8}\, \mathrm du\\[2mm]
\kappa_{d+\sigma} & = \frac{2}{\Gamma\big((d+\sigma+1)/2\big)}  F_{d+\sigma}(\sqrt{2(d+\sigma)}), \qquad \sigma \in \mathbb Z,
\end{align*}
and
\begin{equation}\label{alphad}
F_{d+\sigma}(x) = \frac{\int_x^\infty t^{d+\sigma+1} e^{-t^2/4}\, \mathrm dt}{x^{d+\sigma}e^{-x^2/4}}.
\end{equation}

Introduce
\[
\mathcal B = \left\{ y \in M \ | \ \| y-x\|_m < L h' \log(h'^{-1})^{1/2}   \right\}
\]
with $L>0$ to be chosen
and observe that 
\begin{multline*}
\Delta_{h'}f(x) = \frac{1}{h'^{d+2}(4\pi)^{d/2}} \left( \int_{\mathcal B} + \int_{\mathcal B^c} e^{- \| y-x\|_m^2/4h'^2} \left( f(y) - f(x) \right)
\, \mathrm d \mathrm P(y) \right)\\
= : \mathcal I_{\mathcal B}f(x)+ \mathcal I_{\mathcal B^c}f(x).
\end{multline*} 
We first recall that the distribution $\mathrm P$ has a density $p$ with respect to $\mu.$ Moreover on $\mathcal B$, 
in the $x$-normal coordinates, $\mu$ has a density $\sqrt{\mathrm{det}(g_{ij})}$ and 
the Taylor expansion of $f$ is
\[
f(\mathcal E_x(v)) - f(x) = f(\mathcal E_x(v)) - f(\mathcal E_x(0)) = \langle f'(x), v\rangle + \frac{1}{2} \langle f''(x) v, v \rangle + \frac{1}{6} f'''(\xi)(v,v,v)
\]
for a suitable $\xi= \xi(x) \in M$ and where $f^{(k)}$ denotes the $k$-th derivate with respect to $v$ of $f \circ \mathcal E_x(v)$. Thus we can write
\begin{multline*}
\mathcal I_{\mathcal B}f(x) = 
\frac{1}{h'^{d+2}(4\pi)^{d/2}}  \int_{\mathcal B} e^{- \| y-x\|_m^2/4h'^2} \left( f(y) - f(x) \right)
\, \mathrm d \mathrm P(y) \\
= \frac{1}{h'^{d+2}(4\pi)^{d/2}} \int_{\mathcal E_x^{-1}(\mathcal B)} e^{- \| \mathcal E_x(v) -x\|_m^2/4h'^2} \left(  \langle f'(x), v\rangle + \frac{1}{2} \langle f''(x) v, v \rangle  \right)\ p(\mathcal E_x(v)) \sqrt{\mathrm{det}(g_{ij})}(v) \, \mathrm d v \\
+ \frac{1}{6 h'^{d+2}(4\pi)^{d/2}} \int_{\mathcal E_x^{-1}(\mathcal B)} e^{- \| \mathcal E_x(v) -x\|_m^2/4h'^2}f'''(\xi)(v,v,v)\ p(\mathcal E_x(v)) \sqrt{\mathrm{det}(g_{ij})}(v) \, \mathrm d v.
\end{multline*} 

\noindent
Using now the Taylor expansion of $p$ in $x$-normal coordinates
\[
p(\mathcal E_x(v)) = p(x) + \langle p'(x), v \rangle + \frac{1}{2} \langle p''(\zeta) v, v\rangle
\]
for a suitable $\zeta = \zeta(x) \in M$ and where as before $p^{(k)}$ denotes $k$-th derivate of $p\circ\mathcal E_x$,
we have that 
\begin{multline*}
| \Delta_{ h'}f(x) - p(x) \Delta_{\mathrm P}f(x) | \\
\leq \mathrm I_1f(x) +\mathrm I_2 f(x) +\mathrm I_3 f(x) +\mathrm I_4 f(x) +\mathrm I_5 f(x) +\mathrm I_6 f(x) + \mathcal I_{\mathcal B^c}f(x)
\end{multline*} 
where, denoting by
\begin{equation}\label{st}
\mathcal S(u) =  \langle f'(x), u\rangle \langle p'(x), u\rangle + \frac{p(x)}{2}\langle f''(x) u, u \rangle,
\end{equation}
\begin{align*}
\mathrm I_1f(x) 
& =   \left|  \frac{1}{h'^{d+2}(4\pi)^{d/2}} \int_{\mathcal E_x^{-1}(\mathcal B)} e^{- \| \mathcal E_x(v) -x\|_m^2/4h'^2}  \mathcal S(v)\sqrt{\mathrm{det}(g_{ij})}(v) \, \mathrm d v - p(x) \Delta_{\mathrm P}f(x) \right| \\
\mathrm I_2f(x) 
& = \frac{1}{h'^{d+2}(4\pi)^{d/2}} \left| \int_{\mathcal E_x^{-1}(\mathcal B)} e^{- \| \mathcal E_x(v) -x\|_m^2/4h'^2} p(x) \ \langle f'(x), v\rangle  \sqrt{\mathrm{det}(g_{ij})}(v) \, \mathrm d v\right|\\
\mathrm I_3f(x) 
& = \frac{1}{2h'^{d+2}(4\pi)^{d/2}} \left| \int_{\mathcal E_x^{-1}(\mathcal B)} e^{- \| \mathcal E_x(v) -x\|_m^2/4h'^2} \langle f'(x), v\rangle\  \langle p''(\zeta) v, v\rangle
\sqrt{\mathrm{det}(g_{ij})}(v) \, \mathrm d v\right| \\
\mathrm I_4f(x) 
& =  \frac{1}{2h'^{d+2}(4\pi)^{d/2}} \left| \int_{\mathcal E_x^{-1}(\mathcal B)} e^{- \| \mathcal E_x(v) -x\|_m^2/4h'^2}\langle f''(x) v, v \rangle \langle p'(x), v \rangle
\sqrt{\mathrm{det}(g_{ij})}(v) \, \mathrm d v\right| \\
\mathrm I_5f(x) 
& =  \frac{1}{4h'^{d+2}(4\pi)^{d/2}} \left| \int_{\mathcal E_x^{-1}(\mathcal B)} e^{- \| \mathcal E_x(v) -x\|_m^2/4h'^2}\langle f''(x) v, v \rangle
\langle p''(\zeta) v, v\rangle
\sqrt{\mathrm{det}(g_{ij})}(v) \, \mathrm d v\right|\\
\mathrm I_6f(x) & = \frac{1}{6 h'^{d+2}(4\pi)^{d/2}} \left| \int_{\mathcal E_x^{-1}(\mathcal B)} e^{- \| \mathcal E_x(v) -x\|_m^2/4h'^2}f'''(\xi)(v,v,v)\ p(\mathcal E_x(v)) \sqrt{\mathrm{det}(g_{ij})}(v) \, \mathrm d v\right|.
\end{align*}
Let us now bound each term separately. 
We first observe that by definition
\[
\mathcal I_{\mathcal B^c}f(x)
\leq \frac{2C_{\mathcal F}}{h'^{d+2-L^2/4}(4\pi)^{d/2}}.
\]
For the proofs of ~\cref{lem1hor} and ~\cref{lem2hor}
we refer to ~\cref{sec_proof_hor}.

\vskip2mm
\begin{lem}\label{lem1hor}
It holds
\[
\mathrm I_1f(x) \leq 
 C_{\mathcal F} \left( \|p'\|_\infty + \| p\|_\infty/2 \right) \left[\frac{ \kappa_d} {(4\pi)^{d/2}} L^d \log(h'^{-1})^{d/2} h'^{L^2/4}  + \tilde D_4 \ h'^2 \right].
\]
\end{lem}

\begin{lem}\label{lem2hor}
It holds
\begin{align*}
\mathrm I_2f(x) & \leq C_{\mathcal F} \| p\|_\infty \left[   \frac{\kappa_{d-1}} {(4\pi)^{d/2}} L^{d-1} \log(h'^{-1})^{(d-1)/2} h'^{L^2/4 -1}+\tilde D_3 \ h' \right]\\
\mathrm I_3f(x) & \leq  C_{\mathcal F} \| p''\|_\infty \ h' \left( 5\tau_d + \frac{ \tilde D_5}{2} \ h'^2 \right)\\
\mathrm I_4f(x) & \leq  C_{\mathcal F} \| p'\|_\infty \ h'  \left(5\tau_d+ \frac{ \tilde D_5}{2}  \ h'^2 \right)\\
\mathrm I_5f(x) & \leq C_{\mathcal F} \|p''\|_\infty \ h'^2 \left(  \frac{35}{4}  \tau_d +\frac{ \tilde D_6 }{4} \ h'^2\right)\\
\mathrm I_6f(x) & \leq C_{\mathcal F} \| p\|_\infty h' \left(  \frac{5\tau_d }{3}+  \frac{\tilde D_5}{6} \ h'^2\right).
\end{align*}
\end{lem}
We choose $L= 2\sqrt{d+3}$ and we note that Condition~\ref{CondhmaxRayInj} is then satisfied. This concludes the proof.

\subsection{Proof of ~\cref{cor_c}}\label{proof6}

We first prove the following lemma:
\begin{lem} \label{lem:corfinalbound}
Assume that the density $p$ of $\mathrm P$ is in  $\mathcal C^1(M,\R^+)$. For $f \in \mathcal F$, we have
\begin{multline}\label{cor_a1}
\mathbb E \left[ \| (p\Delta_{\mathrm P}  - \hat \Delta_{\hat h}) f\|_{2,M} ^2 \right] \leq 
 2 \inf_{h \in \mathcal H} \left\{9 D(h)^2 + (1+\sqrt2)^2 {b V(h)}\right\}\\
+ 2  \ C_{\mathcal F}^2 \sum_{h\in \mathcal H} \delta(h)\ \left[ \frac{1}{h_{\text{min}}^{d+2}} \Big( w_d \| K_d\|_{2,d}^2 + \alpha_d(h_{\text{max}})\Big)+ 9  \big(\|p'\|_\infty+\|p\|_\infty/2 \big)^2 \mu(M) \tau_d^2 \right],
\end{multline}
where we recall that
$\tau_d = \frac{1}{(4\pi)^{d/2}} \int_{\mathbb R^d} e^{-\|u\|_d^2/8}\, \mathrm du $.
\end{lem}
\begin{proof}
Let $A$ be the set on which the oracle inequality in 
~\cref{prop_oineq}
holds true. 
Then 
\begin{multline}\label{esp_2eq}
\mathbb E \left[ \| (p\Delta_{\mathrm P}  - \hat \Delta_{\hat h}) f\|_{2,M} ^2 \right] \\
\leq \mathbb E \left[ \| (p\Delta_{\mathrm P}  - \hat \Delta_{\hat h}) f\|_{2,M} ^2 \mathbbm 1_{A}\right] 
+ 2 \mathbb E \left[\left( \| p\Delta_{\mathrm P} f\|_{2,M} ^2 +\| \hat \Delta_{\hat h} f\|_{2,M} ^2\right) \mathbbm 1_{A^c}\right].
\end{multline}
Observe that to bound the first term, it is sufficient to use ~\cref{prop_oineq}. 
Thus we only have to consider the second term. 
According to ~\cref{eqDL}, we can rewrite the weighted Laplace-Beltrami operator
in $x$-normal coordinates as 
\[
\Delta_{\mathrm P}f(x) = \frac{1}{(4\pi)^{d/2}} \int e^{-\|u\|_d^2/4} 
\left( \frac{\langle p'(x), u\rangle}{p(x)} \ \langle f'(x), u \rangle + \frac{1}{2} \langle f''(x)u, u \rangle  \right) \, \mathrm d u,
\]
so that 
\begin{align*}
\| p\Delta_{\mathrm P} f\|_{2,M} ^2
& = \frac{1}{(4\pi)^{d}} \int \left(\int e^{-\|u\|_d^2/4} 
\left( \langle p'(x), u\rangle \ \langle f'(x), u \rangle + \frac{p(x)}{2} \langle f''(x)u, u \rangle  \right) \, \mathrm d u \right)^2 \, \mathrm d \mu(x)\\
& \leq  \frac{C_{\mathcal F}^2 \left(\|p'\|_\infty+\|p\|_\infty/2 \right)^2 \mu(M)}{(4\pi)^{d}}  \left(\int e^{-\|u\|_d^2/4} \|u\|_d^2 \, \mathrm d u \right)^2\\
& \leq  9 C_{\mathcal F}^2 \left(\|p'\|_\infty+\|p\|_\infty/2 \right)^2 \mu(M) \tau_d^2
\end{align*}
where in the last line we have used ~\cref{nb_exp} with $s=2,$ $q=2$ so that $\lambda\leq 3.$
Moreover for any $h \in \mathcal H$ by ~\cref{tec_lem1}
\begin{multline*}
\| \hat \Delta_{h} f\|_{2,M}^2
 \leq \frac{1}{n^2} \left( \sum_{i=1}^n \| H_h(\cdot-X_i) \left( f(X_i) - f(\cdot) \right)\|_{2,M} \right)^2\\
 \leq \mathcal I_2(h)
\leq  \frac{2C_{\mathcal F}^2}{h_{\text{min}}^{d+2}} \Big[ w_d \| K_d\|_{2,d}^2 + \alpha_d(h_{\text{max}})\Big].
\end{multline*}
Hence
\begin{multline*}
\mathbb E \bigg[\Big( \| p\Delta_{\mathrm P} f\|_{2,M} ^2 +\| \hat \Delta_{\hat h} f\|_{2,M} ^2\Big) \mathbbm 1_{A^c}\bigg]\\
\leq  \ C_{\mathcal F}^2 \left[ \frac{1}{h_{\text{min}}^{d+2}} \Big( w_d \| K_d\|_{2,d}^2 + \alpha_d(h_{\text{max}})\Big)+ 9  \big(\|p'\|_\infty+\|p\|_\infty/2 \big)^2 \mu(M) \tau_d^2 \right]\ \sum_{h\in \mathcal H} \delta(h).
\end{multline*}
\end{proof}

We now prove \cref{cor_c}. Note that $ \gamma_d(h') \geq \frac{\mathcal C_{d}}{h'^d}$ with 
$
\mathcal C_{d}
=\frac{1}{\|p\|_\infty } \left[ \frac{ \omega_d \ \|K_d\|_{2,d}^2 + \alpha_d(h_{\text{min}})}{\left(\omega_d\ \| K_d\|_{1,d} + \beta_d(h_{\text{max}})\right)^2}   \right].
$
Thus,
\begin{multline*}
\sum_{h\in \mathcal H} \delta(h)
\leq \sum_{h\in \mathcal H}\sum_{h'\leq h} \max\left\{  \exp\left( -\frac{\min\{\epsilon^2,\epsilon\}   \sqrt n }{24} \right), \,  \exp\left( -
\frac{ \mathcal C_{d} \ c\ \epsilon^2}{h'^d} \right) \right\}\\
\leq |\mathcal H|^2  \exp\left( -\frac{\min\{\epsilon^2,\epsilon\} \ \mathcal C_{d}' }{h^d_{\text{max}}} \right)
\end{multline*}
where $ \mathcal C_{d}'>0$ is a constant
and $|\mathcal H|$ denotes the cardinality of the bandwidth set $\mathcal H$. For
\[
\mathcal H = \left\{ e^{-k} \ , \ \lceil \log\log(n)\rceil \leq k \leq \lfloor \log(n) \rfloor  \right\},
\]
we get that the second term in \cref{cor_a1} is negligible with respect to the first one. Finally, observe that combining ~\cref{prop_Dh} with the definition of the variance term $V(h)$, the optimal trade-off in \cref{lem:corfinalbound} is given by $h \sim n^{-\frac{1}{d+4}}$, which concludes the proof.

\subsection{Proof of ~\cref{tec_lem1}}\label{pf_lem62}

In order to prove ~\cref{tec_eq1}, we consider 
\[
\mathcal B = \mathcal B_x(h) = \{ y \in M \ | \ \| y-x\|_m < L h \log(h^{-1})^{1/2} \} \subset M
\]
where $L>0$ has to be chosen
and we observe that
\[
 \frac{1}{h^{d+2}(4\pi)^{d/2}} \int_{\mathcal B \cup \mathcal B^c} e^{-\|y-x\|_m^2/4h^2} |f(x) - f(y)| \, \mathrm d \mu(y).
\]
We first look at the integral on $\mathcal B$ and we consider the $x$-normal coordinates. 
Taking into account that the measure $\mu$ has a density $\sqrt{\mathrm{det}g_{ij}}$ in the $x$-normal coordinates, we get
\begin{multline*}
 \frac{1}{h^{d+2}(4\pi)^{d/2}} \int_{\mathcal B} e^{-\|y-x\|_m^2/4h^2} |f(x) - f(y)| \, \mathrm d \mu(y)\\
= \frac{1}{h^{d+2}(4\pi)^{d/2}} \int_{\mathcal E_x^{-1}(\mathcal B)} e^{-\|\mathcal E_x(v) -x\|_m^2/4h^2} |f(x) - f(\mathcal E_x(v))| \ \sqrt{\mathrm{det}g_{ij}}(v) \, \mathrm d v\\
\leq \frac{1}{h^{d+2}(4\pi)^{d/2}} \int_{\mathcal E_x^{-1}(\mathcal B)} e^{-\|\mathcal E_x(v) -x\|_m^2/4h^2} |f(x) - f(\mathcal E_x(v))| \ (1+C_1\|v\|_d^2) \, \mathrm d v
\end{multline*}
where in the last line we have used ~\cref{eq_310}.
Using the Taylor expansion of $f$ in $x$-normal coordinates
\begin{equation}\label{taylor1}
f(\mathcal E_x(v)) - f(x) = f(\mathcal E_x(v)) - f(\mathcal E_x(0)) = \langle f'(x), v\rangle + \frac{1}{2} \langle f''(\xi) v, v \rangle 
\end{equation}
where $\xi = \xi(x) \in M$ and $f^(k)$ denotes the $k$-th derivate with respect to $v$ of $f \circ \mathcal E_x(v)$ the above chain of inequalities is bounded by
\begin{multline*}
\frac{1}{h^{d+2}(4\pi)^{d/2}} \Bigg[ \int_{\mathcal E_x^{-1}(\mathcal B)} e^{-\|\mathcal E_x(v) -x\|_m^2/4h^2} |\langle f'(x), v\rangle| \ (1+C_1\|v\|_d^2) \, \mathrm d v\\
+  \frac{1}{2} \int_{\mathcal E_x^{-1}(\mathcal B)} e^{-\|\mathcal E_x(v) -x\|_m^2/4h^2} |\langle f''(\xi) v, v \rangle| \ (1+C_1\|v\|_d^2) \, \mathrm d v
\Bigg]\\
\leq \frac{C_{\mathcal F}}{h^{d+2}(4\pi)^{d/2}}
\Bigg[ \int_{\mathcal E_x^{-1}(\mathcal B)} e^{-\|\mathcal E_x(v) -x\|_m^2/4h^2} \| v\|_d \ (1+C_1\|v\|_d^2) \, \mathrm d v\\
+  \frac{1}{2} \int_{\mathcal E_x^{-1}(\mathcal B)} e^{-\|\mathcal E_x(v) -x\|_m^2/4h^2} \|v\|_d^2 \ (1+C_1\|v\|_d^2) \, \mathrm d v
\Bigg]
\end{multline*}
where in the last line we have used the fact $f \in \mathcal F$ is uniformly bounded up to the third order.
We now consider the two terms separately.
We write
\begin{multline*}
\frac{1}{h^{d+2}(4\pi)^{d/2}} \int_{\mathcal E_x^{-1}(\mathcal B)} e^{-\|\mathcal E_x(v) -x\|_m^2/4h^2} \| v\|_d \ (1+C_1\|v\|_d^2) \, \mathrm d v\\
= \frac{1}{h^{d+2}(4\pi)^{d/2}} \int_{\mathcal E_x^{-1}(\mathcal B)} e^{-\|v\|_d^2/4h^2} \| v\|_d\, \mathrm d v + \mathcal R_1
\end{multline*}
where, according to ~\cref{nb_2},
\begin{multline*}
\mathcal R_1
=  \frac{1}{(4\pi)^{d/2} h^{d+2}} \int_{\mathcal E_x^{-1}(\mathcal B)} \left(e^{-\|\mathcal E_x(v) -x\|_m^2/4h^2} - e^{-\|v\|_d^2/4h^2} \right) \|v\|_d  \, \mathrm d v\\
+  \frac{C_1}{(4\pi)^{d/2} h^{d+2}} \int_{\mathcal E_x^{-1}(\mathcal B)} e^{-\|\mathcal E_x(v) -x\|_m^2/4h^2} \|v\|_d^3  \, \mathrm d v\\
\leq \frac{1}{(4\pi)^{d/2} h^{d+2}} \int_{\mathbb R^d} \left( \frac{C\|v\|_d^5}{4h^2} + C_1\|v\|_d^3 \right) e^{-\|v\|_d^2/8h^2} \, \mathrm d v = h \tilde D_3
\end{multline*}
and by ~\cref{nb_exp}
\begin{multline*}
 \frac{1}{(4\pi)^{d/2} h^{d+2}} \int_{\mathcal E_x^{-1}(\mathcal B)} e^{-\|v\|_d^2 /4h^2} \| v\|_d \, \mathrm d v
\leq \frac{3}{2(4\pi)^{d/2} h^{d+1}} \int_{\mathbb R^d} e^{-\|v\|_d^2 /8h^2}\, \mathrm d v\\
 =  \frac{3 \times 2^{d/2}}{2(4\pi)^{d/2} h} \int_{\mathbb R^d} e^{-\|u\|_d^2 /4}\, \mathrm d v =  \frac{3 \times 2^{d/2-1} \| K_d\|_{1,d}}{h}.
\end{multline*}

\noindent
Similarly, the second term can be written as
\begin{multline*}
\frac{1}{h^{d+2}(4\pi)^{d/2}}
 \int_{\mathcal E_x^{-1}(\mathcal B)} e^{-\|\mathcal E_x(v) -x\|_m^2/4h^2} \|v\|_d^2 \ (1+C_1\|v\|_d^2) \, \mathrm d v\\
 = \frac{1}{h^{d+2}(4\pi)^{d/2}}
  \int_{\mathcal E_x^{-1}(\mathcal B)} e^{-\|v\|_d^2/4h^2} \|v\|_d^2 \, \mathrm d v + \mathcal R_2
\end{multline*}
where, according to ~\cref{nb_2},
\begin{multline*}
\mathcal R_2
=  \frac{1}{(4\pi)^{d/2} h^{d+2}} \int_{\mathcal E_x^{-1}(\mathcal B)} \Big(e^{-\|\mathcal E_x(v) -x\|_m^2/4h^2} -  e^{-\|v\|_d^2/4h^2} \Big) \|v\|_d^2  \, \mathrm d v\\
 +  \frac{C_1}{(4\pi)^{d/2} h^{d+2}} \int_{\mathcal E_x^{-1}(\mathcal B)} e^{-\|\mathcal E_x(v) -x\|_m^2/4h^2} \|v\|_d^4  \, \mathrm d v
 \leq h^2 \tilde D_4
\end{multline*}
and by ~\cref{nb_exp}
\[
\frac{1}{(4\pi)^{d/2} h^{d+2}} \int_{\mathcal E_x^{-1}(\mathcal B)} e^{-\|v\|_d^2 /4h^2} \| v\|_d^2 \, \mathrm d v
\leq 3 \times 2^{d/2} \| K_d\|_{1,d}.
\]
This proves that, on $\mathcal B$,
\begin{multline*}
 \frac{1}{h^{d+2}(4\pi)^{d/2}} \int_{\mathcal B} e^{-\|y-x\|_m^2/4h^2} |f(x) - f(y)| \, \mathrm d \mu(y)\\
 \leq C_{\mathcal F} \left[
 \frac{\omega_d \| K_d\|_{1,d}}{h} + \omega_d \| K_d\|_{1,d} +h \tilde D_3+ \frac{h^2 \tilde D_4}{2}
 \right]
\end{multline*}
where we recall the definition of $\omega_d=3 \times 2^{d/2-1}$ in ~\cref{od}.
We now consider the integral on $\mathcal B^c$. According to the definition of $\mathcal B$
\[
 \frac{1}{h^{d+2}(4\pi)^{d/2}} \int_{\mathcal B} e^{-\|y-x\|_m^2/4h^2} |f(x) - f(y)| \, \mathrm d \mu(y)
\leq \frac{2C_{\mathcal F}\mu(M)}{(4\pi)^{d/2} h^{d+2-L^2/4}}.
\]
Choosing $L=2\sqrt{d+2} $ so that $d+2 -L^2/4 = 0$ we prove ~\cref{tec_eq1}. \\[1mm]
In a similar way we prove ~\cref{tec_eq2}. We consider again the ball $\mathcal B$ defined as above 
and we write
\[
\frac{1}{h^{2d+4}(4\pi)^{d}} \int_{\mathcal B \cup \mathcal B^c} e^{-\|y-x\|_m^2/2h^2} |f(x) - f(y)|^2 \, \mathrm d \mu(y).
\]
Proceeding as above, on $\mathcal B$ we consider the $x$-normal coordinates
and according to ~\cref{eq_310} we get
\begin{multline*}
\frac{1}{h^{2d+4}(4\pi)^{d}} \int_{\mathcal B} e^{-\|y-x\|_m^2/2h^2} |f(x) - f(y)|^2 \, \mathrm d \mu(y)\\
\leq \frac{1}{h^{2d+4}(4\pi)^{d}} \int_{\mathcal E_x^{-1}(\mathcal B)} e^{-\|\mathcal E_x(v) -x\|_m^2/2h^2} |f(x) - f(\mathcal E_x(v))|^2 \ (1+C_1\|v\|_d^2) \, \mathrm d v.
\end{multline*}
Using the Taylor expansion of $f$ in ~\cref{taylor1} we bound the above quantity by
\begin{multline*}
\frac{2}{h^{2d+4}(4\pi)^{d}}
\Bigg[
\int_{\mathcal E_x^{-1}(\mathcal B)} e^{-\|\mathcal E_x(v) -x\|_m^2/2h^2} \langle f'(x), v\rangle^2 \ (1+C_1\|v\|_d^2) \, \mathrm d v\\
+ \frac{1}{4}\int_{\mathcal E_x^{-1}(\mathcal B)} e^{-\|\mathcal E_x(v) -x\|_m^2/2h^2}  \langle f''(\xi) v, v \rangle^2 \ (1+C_1\|v\|_d^2) \, \mathrm d v
\Bigg]\\
\leq \frac{2C_{\mathcal F}^2}{h^{2d+4}(4\pi)^{d}}
\Bigg[\int_{\mathcal E_x^{-1}(\mathcal B)} e^{-\|\mathcal E_x(v) -x\|_m^2/2h^2} \|v\|^2 \ (1+C_1\|v\|_d^2) \, \mathrm d v\\
+ \frac{1}{4}\int_{\mathcal E_x^{-1}(\mathcal B)} e^{-\|\mathcal E_x(v) -x\|_m^2/2h^2}  \| v\|^4 \ (1+C_1\|v\|_d^2) \, \mathrm d v
\Bigg]
\end{multline*}
where in the last line we have used the fact that $f \in \mathcal F$. 
We observe that the first term can be written as
\begin{multline*}
\frac{1}{h^{2d+4}(4\pi)^{d}}\int_{\mathcal E_x^{-1}(\mathcal B)} e^{-\|\mathcal E_x(v) -x\|_m^2/2h^2} \|v\|^2 \ (1+C_1\|v\|_d^2) \, \mathrm d v\\
= \frac{1}{h^{2d+4}(4\pi)^{d}}\int_{\mathcal E_x^{-1}(\mathcal B)} e^{-\|v\|_d^2/2h^2} \|v\|^2 \, \mathrm d v + \widetilde {\mathcal R}_1
\end{multline*}
where, by ~\cref{nb_2},
\begin{multline*}
\widetilde {\mathcal R}_1 
= \frac{1}{(4\pi)^{d} h^{2d+4}} \int_{\mathcal E_x^{-1}(\mathcal B)} \left( e^{-\|\mathcal E_x(v)- x\|_m^2/2h^2} - e^{- \| v\|_d^2/2h^2} \right)\ \| v\|_d^2   \, \mathrm d v\\
 +  \frac{C_1}{(4\pi)^{d} h^{2d+4}} \int_{\mathcal E_x^{-1}(\mathcal B)} e^{-\|\mathcal E_x(v)- x\|_m^2/2h^2} \| v\|_d^4 \, \mathrm d v\\
 \leq  \frac{1}{(4\pi)^d h^{2d+4}} \int_{\mathcal E_x^{-1}(\mathcal B)} \left( \frac{C \|v\|_d^6}{2h^2} + C_1  \|v\|_d^4 \right) e^{-\|v\|_d^2/4h^2}\, \mathrm d v
\leq \frac{D_4}{h^d}
\end{multline*}
and according to ~\cref{nb_exp} with $s^2 = 2h^2$ and $q=2$ so that $\lambda \leq \frac{3h^2}{2}$, 
\begin{multline*}
\frac{1}{(4\pi)^d h^{2d+4}} \int_{\mathcal E_x^{-1}(\mathcal B)} e^{-\|v\|_d^2/2h^2}\ \|v\|_d^2 \, \mathrm d v
 \leq \frac{3}{2(4\pi)^d h^{2d+2}} \int_{\mathbb R^d} e^{-\|v\|_d^2/4h^2} \, \mathrm d v\\
 = \frac{3\times 2^{d/2}}{2(4\pi)^d h^{d+2}} \int_{\mathbb R^d} e^{-\|u\|_d^2/2} \, \mathrm d u = \frac{3\times2^{d/2-1}\|K_d\|_{2,d}^2}{h^{d+2}}.
\end{multline*}
Moreover the second term writes
\begin{multline*}
\frac{1}{h^{2d+4}(4\pi)^{d}}\int_{\mathcal E_x^{-1}(\mathcal B)} e^{-\|\mathcal E_x(v) -x\|_m^2/2h^2} \|v\|^4 \ (1+C_1\|v\|_d^2) \, \mathrm d v\\
= \frac{1}{h^{2d+4}(4\pi)^{d}}\int_{\mathcal E_x^{-1}(\mathcal B)} e^{-\|v\|_d^2/2h^2} \|v\|^4 \, \mathrm d v + \widetilde {\mathcal R}_2
\end{multline*}
where, using again ~\cref{nb_2},
\begin{multline*}
\widetilde {\mathcal R}_2 = 
\frac{1}{(4\pi)^{d} h^{2d+4}} \int_{\mathcal E_x^{-1}(\mathcal B)} \left( e^{-\|\mathcal E_x(v)- x\|_m^2/2h^2} - e^{- \| v\|_d^2/2h^2} \right)\ \| v\|_d^4   \, \mathrm d v\\
 +  \frac{C_1}{(4\pi)^{d} h^{2d+4}} \int_{\mathcal E_x^{-1}(\mathcal B)} e^{-\|\mathcal E_x(v)- x\|_m^2/2h^2} \| v\|_d^6 \, \mathrm d v\\
 \leq \frac{1}{(4\pi)^d h^{2d+4}} \int_{\mathbb R^d} \left( \frac{C\|v\|_d^8}{2h^2} + C_1\|v\|_d^6\right) e^{-\|v\|_d^2/4h^2} \, \mathrm d v
= \frac{D_6}{h^{d-2}}
\end{multline*}
and by ~\cref{nb_exp}
with $\lambda \leq 9h^4$,
\[
 \frac{1}{(4\pi)^d h^{2d+4}} \int_{\mathcal E_x^{-1}(\mathcal B)} e^{-\|v\|_d^2/2h^2} \| v\|_d^4\, \mathrm dv
 \leq \frac{9\times 2^{d/2} \|K_d\|_{2,d}^2}{ h^d}.
\]
Then on $\mathcal B$ we have
\begin{multline*}
\frac{1}{h^{2d+4}(4\pi)^{d}} \int_{\mathcal B} e^{-\|y-x\|_m^2/2h^2} |f(x) - f(y)|^2 \, \mathrm d \mu(y)\\
\leq 2C_{\mathcal F}^2 \left[
\frac{\omega_d  \|K_d\|_{2,d}^2}{h^{d+2}}+ \frac{D_4+ 3\omega_d \|K_d\|_{2,d}^2/2}{h^d} + \frac{D_6}{4h^{d-2}}
\right]
\end{multline*}
while on $\mathcal B^c$ 
\[
\frac{1}{h^{2d+4}(4\pi)^{d}} \int_{\mathcal B^c} e^{-\|y-x\|_m^2/2h^2} |f(x) - f(y)|^2 \, \mathrm d \mu(y)
\leq \frac{4C_{\mathcal F}^2\mu(M)}{(4\pi)^{d}h^{2d+4-L^2/2}}.
\]
Thus choosing $L= \sqrt{2(d+4)}$ so that $ 2d+4-L^2/2= d$ we conclude the proof.

\subsection{Proofs of the technical lemmas in ~\cref{proof7}}\label{sec_proof_hor}

\subsubsection{Proof of ~\cref{lem1hor}}

According to ~\cref{eq_310}, 
\begin{multline*}
\mathrm I_1f(x) 
\leq  \left|  \frac{1}{h'^{d+2}(4\pi)^{d/2}} \int_{\mathcal E_x^{-1}(\mathcal B)} e^{- \| \mathcal E_x(v) -x\|_m^2/4h'^2}  \mathcal S(v) \, \mathrm d v - p(x) \Delta_{\mathrm P}f(x) \right| \\
+ \frac{C_1}{h'^{d+2}(4\pi)^{d/2}} \int_{\mathcal E_x^{-1}(\mathcal B)} e^{- \| \mathcal E_x(v) -x\|_m^2/4h'^2}  \mathcal S(v) \|v\|_d^2 \, \mathrm d v\\
\leq   \left|  \frac{1}{h'^{d+2}(4\pi)^{d/2}} \int_{\mathcal E_x^{-1}(\mathcal B)} e^{- \|v\|_d^2/4h'^2}  \mathcal S(v) \, \mathrm d v - p(x) \Delta_{\mathrm P}f(x) \right| + \mathcal R_1
\end{multline*}
where
\begin{multline*}
\mathcal R_1 
\leq \frac{1}{h'^{d+2}(4\pi)^{d/2}}\left| \int_{\mathcal E_x^{-1}(\mathcal B)} \left( e^{- \| \mathcal E_x(v) -x\|_m^2/4h'^2} - e^{- \|v\|_d^2/4h'^2}  \right) \mathcal S(v)\, \mathrm d v\right|\\
+ \frac{C_1}{h'^{d+2}(4\pi)^{d/2}}\left| \int_{\mathcal E_x^{-1}(\mathcal B)}e^{- \| \mathcal E_x(v) -x\|_m^2/4h'^2} \| v\|_d^2\ \mathcal S(v)\, \mathrm d v\right|.
\end{multline*}
Since
\[ 
| \mathcal S(t) | \leq C_{\mathcal F} \left( \|p'\|_\infty + \| p\|_\infty/2 \right) \|t \|_d^2
\]
by  ~\cref{nb_exp}
\begin{multline*}
\mathcal R_1 
\leq \frac{ C_{\mathcal F} \left( \|p'\|_\infty + \| p\|_\infty/2 \right)}{h'^{d+2}(4\pi)^{d/2}}
 \int_{\mathcal E_x^{-1}(\mathcal B)} e^{- \|v\|_d^2/8h'^2} \| v\|_d^2 \left( \frac{C\|v\|_d^4}{4h'^2}+ C_1\|v\|_d^2\right)  \, \mathrm d v\\
\leq C_{\mathcal F} \left( \|p'\|_\infty + \| p\|_\infty/2 \right) \tilde D_4 \ h'^2.
\end{multline*} 
Moreover according to ~\cref{eqDL} we can rewrite the weighted Laplace-Beltrami operator
in $x$-normal coordinates as 
\[
\Delta_{\mathrm P}f(x) = \frac{1}{(4\pi)^{d/2}} \int e^{-\|u\|_d^2/4} 
\left( \frac{\langle p'(x), u\rangle}{p(x)} \ \langle f'(x), u \rangle + \frac{1}{2} \langle f''(x)u, u \rangle  \right) \, \mathrm d u,
\]
so that
\[
p(x) \Delta_{\mathrm P}f(x) = \frac{1}{(4\pi)^{d/2}} \int e^{-\|u\|_d^2/4} \mathcal S(u) \, \mathrm d u.
\]
Hence, denoting 
\[
\tilde{\mathcal B} :=\left\{ v \in \mathbb R^d \ | \ \| v\|_d < L h' \log(h'^{-1})^{1/2}  \right\}
\subset \mathcal E_x^{-1}(\mathcal B). 
\]
we get
\begin{align*}
\bigg|  \frac{1}{h'^{d+2}(4\pi)^{d/2}}& \int_{\mathcal E_x^{-1}(\mathcal B)} e^{- \|v\|_d^2/4h'^2}  \mathcal S(v) \, \mathrm d v - p(x) \Delta_{\mathrm P}f(x) \bigg|\\
& = \frac{1}{h'^{d+2}(4\pi)^{d/2}}  \left| \int_{\left(\mathcal E_x^{-1}(\mathcal B)\right)^c} e^{- \|v\|_d^2/4h'^2}  \mathcal S(v) \, \mathrm d v\right| \\
& \leq  \frac{C_{\mathcal F} \left( \|p'\|_\infty + \| p\|_\infty/2 \right)}{h'^{d+2}(4\pi)^{d/2}} \int_{\tilde {\mathcal B}^c} e^{- \|v\|_d^2/4h'^2}  \|v\|_d^2 \, \mathrm d v\\
& \leq \frac{C_{\mathcal F} \left( \|p'\|_\infty + \| p\|_\infty/2 \right)}{(4\pi)^{d/2}} \int_{\left\{ \| u\|_d > L \log(h'^{-1})^{1/2}  \right\}} e^{- \|u\|_d^2/4}  \|u\|_d^2 \, \mathrm d u\\
& \leq \frac{ { \kappa_d} \ C_{\mathcal F} \left( \|p'\|_\infty + \| p\|_\infty/2 \right)}{(4\pi)^{d/2}} L^d \log(h'^{-1})^{d/2} h'^{L^2/4}
\end{align*}
which concludes the proof.

\subsubsection{Proof of ~\cref{lem2hor}}

By ~\cref{eq_310}
\begin{multline*}
\mathrm I_2f(x) 
\leq  \frac{\|p\|_\infty}{h'^{d+2}(4\pi)^{d/2}}  \left| \int_{\mathcal E_x^{-1}(\mathcal B)} e^{- \| \mathcal E_x(v) - x\|_m^2/4h'^2}
\langle f'(x), v\rangle \left(1+C_1\|v\|_d^2\right) \, \mathrm d v\right|  \\
\leq \frac{\|p\|_\infty}{h'^{d+2}(4\pi)^{d/2}}  \left| \int_{\mathcal E_x^{-1}(\mathcal B)} e^{- \| v\|_d^2/4h'^2}
\langle f'(x), v\rangle\, \mathrm d v\right| + \mathcal R_2
\end{multline*}
where
\begin{multline*}
\mathcal R_2 
=  \frac{\|p\|_\infty}{h'^{d+2}(4\pi)^{d/2}}  \Bigg[
\int_{\mathcal E_x^{-1}(\mathcal B)} \left( e^{- \| \mathcal E_x(v) - x\|_m^2/4h'^2} - e^{- \| v\|_d^2/4h'^2}\right) \langle f'(x), v\rangle\, \mathrm d v\\
+ C_1\int_{\mathcal E_x^{-1}(\mathcal B)} e^{- \| \mathcal E_x(v) - x\|_m^2/4h'^2} \|v\|_d^2  \ \langle f'(x), v\rangle\, \mathrm d v
\Bigg]\\
\end{multline*}
Observe that,  by ~\cref{nb_exp} and ~\cref{eq_310},
\begin{multline*}
\mathcal R_2 
\leq \frac{C_{\mathcal F}\|p\|_\infty}{h'^{d+2}(4\pi)^{d/2}}  \Bigg[
\int_{\mathcal E_x^{-1}(\mathcal B)} \left( e^{- \| \mathcal E_x(v) - x\|_m^2/4h'^2} - e^{- \| v\|_d^2/4h'^2}\right) \| v\|_d\, \mathrm d v\\
+ C_1\int_{\mathcal E_x^{-1}(\mathcal B)} e^{- \| \mathcal E_x(v) - x\|_m^2/4h'^2} \|v\|_d^3 \, \mathrm d v
\Bigg]\\
\leq \frac{C_{\mathcal F} \| p\|_\infty}{h'^{d+2}(4\pi)^{d/2}} \int_{\mathcal E_x^{-1}(\mathcal B)} \| v\|_d \ e^{- \| v\|_d^2/8h'^2}
\left(\frac{C \|v\|_d^4}{4h'^2} + C_1 \|v\|_d^2 \right)  \, \mathrm d v\\
\leq C_{\mathcal F} \| p\|_\infty \tilde D_3 \ h'
\end{multline*}
where in the last line we have used the definition of $\tilde D_\alpha$ introduced in ~\cref{def_tDa}.
Moreover, observe that 
\[
\int_{\mathbb R^d} e^{- \| v\|_d^2/4h'^2} \ \langle f'(x), v\rangle  \, \mathrm d v =0
\]
and define
\[
\tilde{\mathcal B} :=\left\{ v \in \mathbb R^d \ | \ \| v\|_d < L h' \log(h'^{-1})^{1/2}  \right\}
\subset \mathcal E_x^{-1}(\mathcal B). 
\]
Then
we get
\begin{multline*}
\frac{\|p\|_\infty}{h'^{d+2}(4\pi)^{d/2}}  \left| \int_{\mathcal E_x^{-1}(\mathcal B)} e^{- \| v\|_d^2/4h'^2}
\langle f'(x), v\rangle\, \mathrm d v\right| \\
= \frac{\|p\|_\infty}{h'^{d+2}(4\pi)^{d/2}}  \left| \int_{\mathcal E_x^{-1}(\mathcal B)^c} e^{- \| v\|_d^2/4h'^2}
\langle f'(x), v\rangle\, \mathrm d v\right| \\
\leq \frac{\|p\|_\infty}{h'^{d+2}(4\pi)^{d/2}}  \left| \int_{\tilde{\mathcal B}^c} e^{- \| v\|_d^2/4h'^2}
\langle f'(x), v\rangle\, \mathrm d v\right| \\\leq  \frac{C_{\mathcal F} \| p\|_\infty}{h'^{d+2}(4\pi)^{d/2}} \int_{\tilde{\mathcal B}^c} e^{- \| v\|_d^2/4h'^2} \|v\|_d \, \mathrm d v\\
\leq  \frac{C_{\mathcal F} \| p\|_\infty}{h'(4\pi)^{d/2}} \int_{\left\{ \| u\|_d > L \log(h'^{-1})^{1/2}  \right\}} e^{- \| u\|_d^2/4} \|u\|_d \, \mathrm d v\\
\leq  \frac{ \kappa_{d-1} \ C_{\mathcal F} \| p\|_\infty}{(4\pi)^{d/2}} L^{d-1} \log(h'^{-1})^{(d-1)/2} h'^{L^2/4 -1}.
\end{multline*}
Similarly we can bound $\mathrm I_3f(x)$. By ~\cref{eq_310} we have
\begin{multline*}
\mathrm I_3f(x) 
\leq  \frac{C_{\mathcal F} \| p''\|_\infty}{2h'^{d+2}(4\pi)^{d/2}} 
 \int_{\mathcal E_x^{-1}(\mathcal B)} e^{- \| \mathcal E_x(v) -x\|_m^2/4h'^2}\|v\|_d^3 \left( 1+ C_1\|v\|_d^2\right)  \, \mathrm d v\\
 =\frac{C_{\mathcal F} \| p''\|_\infty}{2h'^{d+2}(4\pi)^{d/2}} \int_{\mathcal E_x^{-1}(\mathcal B)} e^{- \| v\|_d^2/4h'^2}\|v\|_d^3  \, \mathrm d v+ \mathcal R_3
 \end{multline*}
 where, according to  ~\cref{nb_exp},
\begin{multline*}
\mathcal R_3 
= \frac{C_{\mathcal F} \| p''\|_\infty}{2h'^{d+2}(4\pi)^{d/2}} 
\Bigg[  \int_{\mathcal E_x^{-1}(\mathcal B)} \left(e^{- \| \mathcal E_x(v) -x\|_m^2/4h'^2} - e^{- \| v\|_d^2/4h'^2}\right) \|v\|_d^3  \, \mathrm d v\\
+ C_1 \int_{\mathcal E_x^{-1}(\mathcal B)} e^{- \| \mathcal E_x(v) -x\|_m^2/4h'^2}  \|v\|_d^5  \, \mathrm d v\Bigg]\\
\leq \frac{1}{2} C_{\mathcal F} \| p''\|_\infty \tilde D_5 \ h'^3
\end{multline*}
and by ~\cref{nb_exp} with $s= 2h'$ and $q=3$ so that $\lambda \leq 10h'^3$
\[
\frac{C_{\mathcal F} \| p''\|_\infty}{2h'^{d+2}(4\pi)^{d/2}} \int_{\mathcal E_x^{-1}(\mathcal B)} e^{- \| v\|_d^2/4h'^2}\|v\|_d^3  \, \mathrm d v
\leq 5 C_{\mathcal F} \| p''\|_\infty \tau_d\ h'.
\]
The bounds for $\mathrm I_4f(x), \ \mathrm I_5f(x), \mathrm I_6f(x)$ are obtained in the same way.

\end{document}